\documentclass[preprint]{imsart}

\RequirePackage[OT1]{fontenc}
\RequirePackage{amsthm,amsmath}
\RequirePackage{natbib}
\RequirePackage[colorlinks,citecolor=blue,urlcolor=blue]{hyperref}

\firstpage{0}
\lastpage{0}

\startlocaldefs
\usepackage{graphicx,psfrag,epsf,epstopdf}
\usepackage{xcolor}

\newtheorem{theorem}{Theorem}
\newtheorem{lemma}{Lemma}

\theoremstyle{definition}

\newtheorem{remark}{Remark}
\newtheorem{assumption}{Assumption}

\newcommand{\bX}{{\bf X}}

\newcommand{\bep}{\mbox{\boldmath $\varepsilon$}}
\newcommand{\bbeta}{\boldsymbol\beta}

\newcommand{\be}{\begin{equation}}
\newcommand{\ee}{\end{equation}}
\newcommand{\bee}{\begin{equation*}}
\newcommand{\eee}{\end{equation*}}

\definecolor{red}{rgb}{0,0,0}

\DeclareMathOperator{\sign}{sign}

\endlocaldefs

\begin{document}

\begin{frontmatter}

\title{Selection by Partitioning the Solution Paths}
\runtitle{Selection by Partitioning the Solution Paths}


\begin{aug}
	
\author{\fnms{Yang} \snm{Liu}\ead[label=e1]{yliu23@fhcrc.org}}
\address{Public Health Science Division, Fred Hutchinson Cancer Research Center,\\
	Seattle, WA, 98109\\
	\printead{e1}}

\

\and 

\author{\fnms{Peng} \snm{Wang}
	\ead[label=e2]{wangp9@ucmail.uc.edu}}

\address{Department of Operations, Business Analytics and Information Systems, University of Cincinnati, Cincinnati, OH, 45220  \\
	\printead{e2}}

\runauthor{Y. Liu and P. Wang}

\affiliation{University of Cincinnati}

\end{aug}

\begin{abstract}
The performance of penalized likelihood approaches depends profoundly on the selection of the tuning parameter; however, there is no commonly agreed-upon criterion for choosing the tuning parameter. Moreover, penalized likelihood estimation based on a single value of the tuning parameter suffers from several drawbacks. This article introduces a novel approach for feature selection based on the entire solution paths rather than the choice of a single tuning parameter, which significantly improves the accuracy of the selection. Moreover, the approach allows for feature selection using ridge or other strictly convex penalties. The key idea is to classify variables as relevant or irrelevant at each tuning parameter and then to select all of the variables which have been classified as relevant at least once. We establish the theoretical properties of the method, which requires significantly weaker conditions than existing methods in the literature. We also illustrate the advantages of the proposed approach with simulation studies and a data example.
\end{abstract}


\tableofcontents
\end{frontmatter}

\section{Introduction}
\label{sec:intro}

The penalized likelihood approach has been very popular for feature selection problems. Under the currently used framework, one first needs to compute the solution paths and then choose a tuning parameter based on a certain criterion. The solution yielded with the chosen tuning parameter is considered to provide the final estimates of the parameters. The problem of choosing a proper tuning parameter is notoriously difficult. In this paper, we propose to select features by utilizing the entire solution paths rather than searching for a single value of the tuning parameter.    

Consider the following linear regression model:
\be
\label{lm}
\mathbf{y}=\mathbf{X}\boldsymbol{\beta^\ast}+\bep;\  \bep\sim\mathcal{N}(\mathbf{0},\sigma^2\mathbf{I}),
\ee
where $\mathbf{y}=(y_1,\dots,y_n)^T$ is an $n$-dimensional response vector, $\boldsymbol{\beta^\ast}=(\beta_1^\ast,\dots,\beta_p^\ast)^T$ is a $p$-dimensional vector of regression coefficients, $\mathbf{X}=(\mathbf{x_1},\dots,\mathbf{x_p})$ is an $n\times p$ design matrix, and $\bep=(\epsilon_1,\dots,\epsilon_n)^T$ is an $n$-dimensional vector of independent and identically distributed random errors. We assume that all of the variables in $\mathbf{X}$ are standardized, so that the coefficients in $\boldsymbol{\beta^\ast}$ are on the same scale.  For the linear regression problems described in \eqref{lm}, the penalized likelihood approach is equivalent to the penalized least squares regression, and the regression coefficients are estimated by minimizing the following objective function:
\be
\label{penalty}
\frac{1}{n}\|\mathbf{y}-\boldsymbol{X\beta}\|_2^2+\lambda\sum_{j=1}^p J(|\beta_j|),
\ee
where $J(\cdot)$ is a penalty function which controls the number of non-zero coefficients and $\lambda>0$ is a tuning parameter. For the penalty function $J(\cdot)$,  one could use a convex penalty like the lasso \citep{tibshirani1996regression} or the adaptive lasso \citep{zou2006adaptive} or one of the non-convex penalties such as the smoothly clipped absolute deviation (SCAD) penalty \citep{fan2001variable}, the minimax concave penalty (MCP) \citep{zhang2010regularization}, or the truncated $l_1$ penalty   \citep{shen2012likelihood}, among others.

Some general selection criteria for the tuning parameter $\lambda$ include cross-validation (CV), generalized cross-validation (GCV), the Akaike information criterion (AIC) \citep{akaike_information_1973}, the Bayesian information criterion (BIC) \citep{schwarz_estimating_1978}, and the generalized information criterion \citep{atkinson_note_1980}. \citet{chen2008extended} pointed out that these criteria usually identify too many irrelevant features when the number of variables is large. This phenomenon has also been described by \citet{broman2002model}, \citet{siegmund2004model}, and \citet{bogdan2004modifying} in their studies of quantitative loci mapping. \citet{chen2008extended} proposed the extended BIC (EBIC), which promotes model sparsity by adjusting BIC with an additional penalty term for the growing number of parameters in the model. Recently, \citet{sun2013consistent}  developed a new technique via variable selection stability, which directly focuses on the selection of informative variables. \citet{fan2013tuning} also proposed to select the tuning parameter by optimizing the generalized information criterion with an appropriate penalty depending on the dimensions in the model. 

Although the above criteria have been well studied for over a decade, there is no concurrence of opinion regarding which criterion to employ for choosing the tuning parameter. As examples,  see Table \ref{criteria} for a short list of publications in major statistics and machine learning journals and the different criteria they use. In fact, the feature selection procedure currently in use, which utilizes only one chosen value for the tuning parameter, suffers from the unavoidable drawback that it is often impossible to correctly identify all the features, no matter which criterion is used.


\begin{table}
	\caption{The general selection criteria in the statistics and machine learning literature in major journals.}
	\label{criteria}
	\centerline{
			\begin{tabular}{lll}
				\hline
				Criterion & Method &References \\
				\hline
				CV        & Adaptive lasso &\citet{zou2006adaptive}, JASA \\
				& Fused lasso &\citet{tibshirani2005sparsity}, JRSSB\\
				& Truncated $l_1$ penalty &\citet{shen2012likelihood}, JASA \\
				\hline
				GCV       & Lasso &\citet{tibshirani1996regression}, JRSSB \\
				& SCAD &\citet{fan2001variable}, JASA \\
				\hline
				BIC       & Lasso &\citet{wang2007regression}, JRSSB         \\
				& Lasso&\citet{yuan2007model}, Biometrika       \\
				\hline
				EBIC      &  Tilting &\citet{cho2012high}, JRSSB \\
				&  Group lasso &\citet{huang2010variable}, AOS \\
				\hline
	\end{tabular}}
\end{table}
{
 \color{red}
 To overcome this drawback, we develop an innovative, intuitive approach which avoids choosing a tuning parameter. Instead, our proposed approach utilizes information from the entire solution paths to improve the selection accuracy. Thus we named the proposed procedure  {\it selection by partitioning the solution paths} (SPSP). Besides improvement in selection accuracy, the SPSP approach can also achieve selection consistency under a weaker condition compared to the currently used framework. This is because SPSP does not require selection consistency at any specific values of the tuning parameter.  As a matter of fact, we do not even require the estimation of coefficients shrunk to zeros. This brings another major advantage of the SPSP approach:  we can achieve feature selection with an $l_2$ penalty or ridge regression. It is well known that ridge regression is more stable and handles collinearity better compared to the lasso. Now with the SPSP approach, we can enjoy these nice properties of the ridge regression without sacrificing capabilities of feature selection.  Moreover, the minimizer of $l_2$ penalized likelihood function often has an explicit solution. Therefore SPSP also greatly reduces the efforts in algorithm development in presence of complicated likelihood functions. 
}   
%
%
%

{\color{red}
%
	
The SPSP procedure is related to the stability selection approach proposed in the seminal discussion paper by \citet{meinshausen2010stability}. Their approach is based on the probabilities that variables will be selected, and the probabilities are obtained from a generic sub-sampling approach. Therefore, the stability selection approach does not require the selection of a tuning parameter either. However, unlike the proposed SPSP procedure, stability selection does not work with ridge regression. Moreover, the computational cost of the SPSP procedure is only a tiny fraction of that for stability selection, since no sub-sampling is involved. Finally, we find from simulation studies that stability selection tends to select too few variables and therefore produces a significantly higher false negative rate than SPSP. This work is also remotely related to Bayesian variable selection approaches, where the tuning parameters or candidate models are assigned a prior distributions and the posterior distributions of the models are evaluated. See, for example, \citet{hoeting1999bayesian, raftery1997bayesian, posada2004model, barbieri2004optimal}.  }

The rest of this article is organized as follows. {\color{red} Section 2 provides a simulated example to motivate the problem. }
Section 3 introduces the SPSP approach. Section 4 discusses the selection consistency  of the SPSP procedure. Section 5 presents the results from various simulation examples. Section 6 provides an application of SPSP in a cancer study to detect the significant genes for glioblastomas, and Section 7 discusses the advantages and future potential of this work. Appendix provides technique proofs and additional simulation studies.

\section{Selection by Partitioning the Solution Paths}

In this section, we will illustrate the advantage and the idea of the proposed SPSP procedure with an example from model \eqref{lm} in \citet{wang2011random}, where 
$p=40$, 
$\beta_1^\ast=\cdots=\beta_5^\ast=3,$ and $\beta_6^\ast=\cdots=\beta_{10}^\ast=-2$. The rest of the coefficients $\beta_j^\ast, j=11, \dots, 40$, are all zeros. 
The entries for the variables $\mathbf{x}_{j},j=1,\dots,p$, are generated from a multivariate Gaussian distribution whose marginals are the standard normal distribution. The pairwise correlation between the first $10$ variables $\mathbf{x}_{j},j=1,\dots,10$, is $0.9$. The remaining $30$ variables  $\mathbf{x}_{j},j=11,\dots,40$, are mutually independent and they are also independent of the first $10$ variables. Furthermore, we generate errors from the normal distribution with mean 0 and standard deviation 3.  The sample size is set to be $n=50$.

We compute the solution paths and plot them in Figure~\ref{solution-path}.
The dashed lines represent the solution paths for the non-zero coefficients and the solid lines represent those for the zero ones. The tuning parameters chosen by the twofold cross-validation,  generalized cross-validation , AIC, BIC, and the extended BIC are shown by the vertical lines. 
Here, ``BEST" is the tuning parameter  which minimizes the number of incorrect selections (false positives + false negatives). 
We observe that cross-validation, generalized cross-validation, AIC, and BIC tend to select too many spurious variables and that the extended BIC tends to drop most of the important variables. Even for the ``BEST'' selection, the model excludes many non-zero coefficients.
The problem is more evident when we focus on the three lines in the right panel of Figure \ref{solution-path}. Here the two dotted lines (1 and 3) are the solution paths for two non-zero coefficients,  while the solid line (2) is the solution path for a zero coefficient.
Apparently, selecting a small $\lambda$, as AIC, BIC, and generalized cross-validation do, misleads us into identifying all three coefficients as non-zero. On the other hand, a large $\lambda$, as selected in cross-validation, the extended BIC, and ``BEST", incorrectly shrinks both coefficients of the important variables to zero. In fact, it is impossible to correctly identify all three features no matter which value of the tuning parameter we choose, although one can likely see the differences between the three features by visual inspection. 
{
	\color{red}
	For this example, 
	we also provide the solution paths of elastic-net in the left panel of Figure~\ref{af5}, which shows that these traditional tuning parameter criteria also suffer from the same problem for elastic-net.
}

\begin{figure}[htbp]
	\vspace{6pc}
	\begin{center}
		\includegraphics[width=0.49\columnwidth]{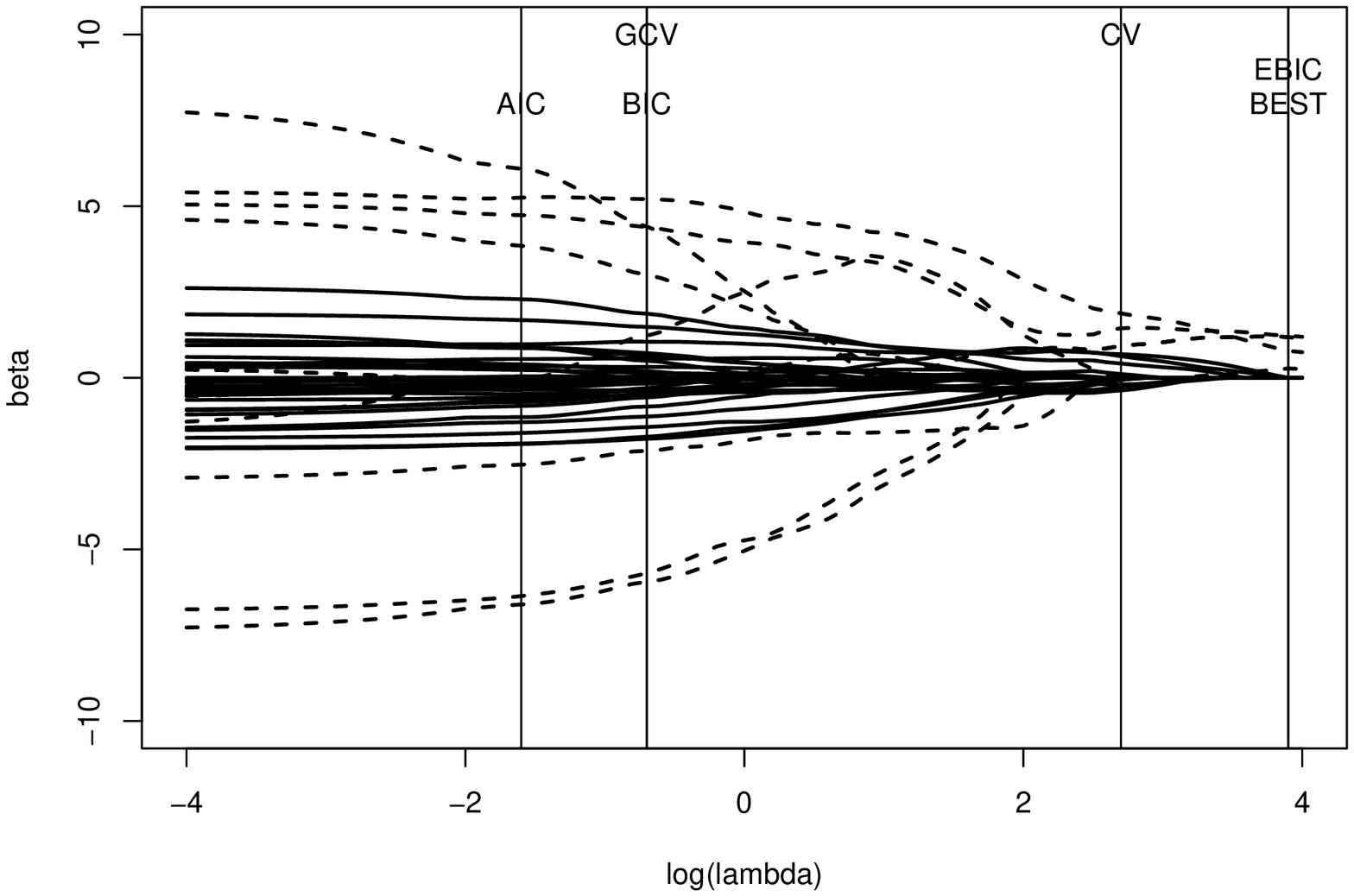}
		\includegraphics[width=0.49\columnwidth]{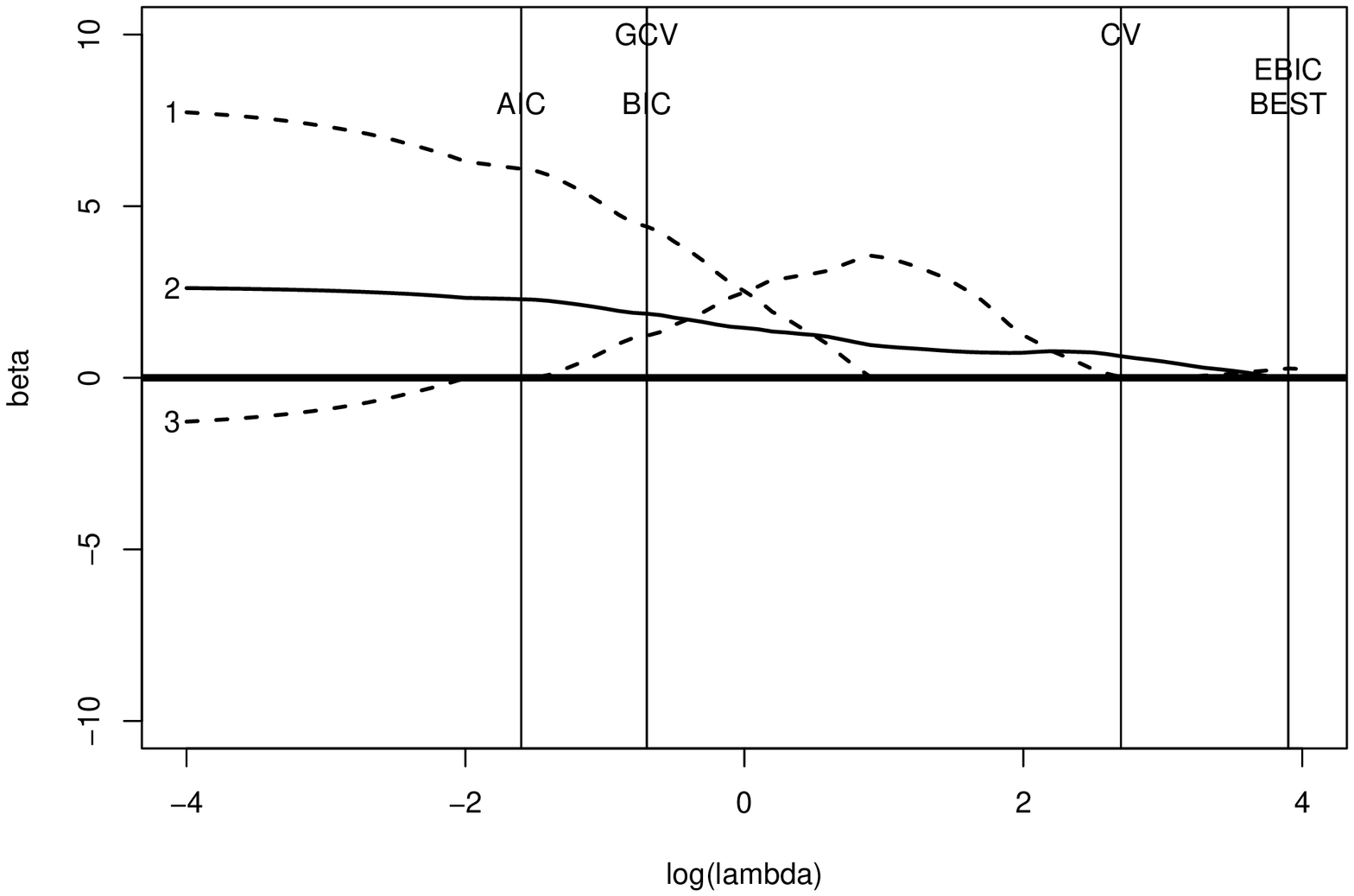}
		\caption[Left: The lasso solution paths for the simulated example. Right: The lasso solution paths of the non-zero variables ``1'', ``2'', and the zero variable ``3''.]{Left: The lasso solution paths for the simulated example. The dashed lines are the paths of the $10$ non-zero coefficients, while the black lines are the paths of the $30$ zero coefficients 
			The vertical lines represent  the tuning parameters selected by different criteria.
			Right: The lasso solution paths for the non-zero coefficients, 1 and {\color{red}3}, and the zero coefficient, {\color{red}2}. Here CV is cross-validation, GCV is generalized cross-validation and EBIC is extended BIC.
		}
		\label{solution-path}
		
		\includegraphics[width=0.49\columnwidth]{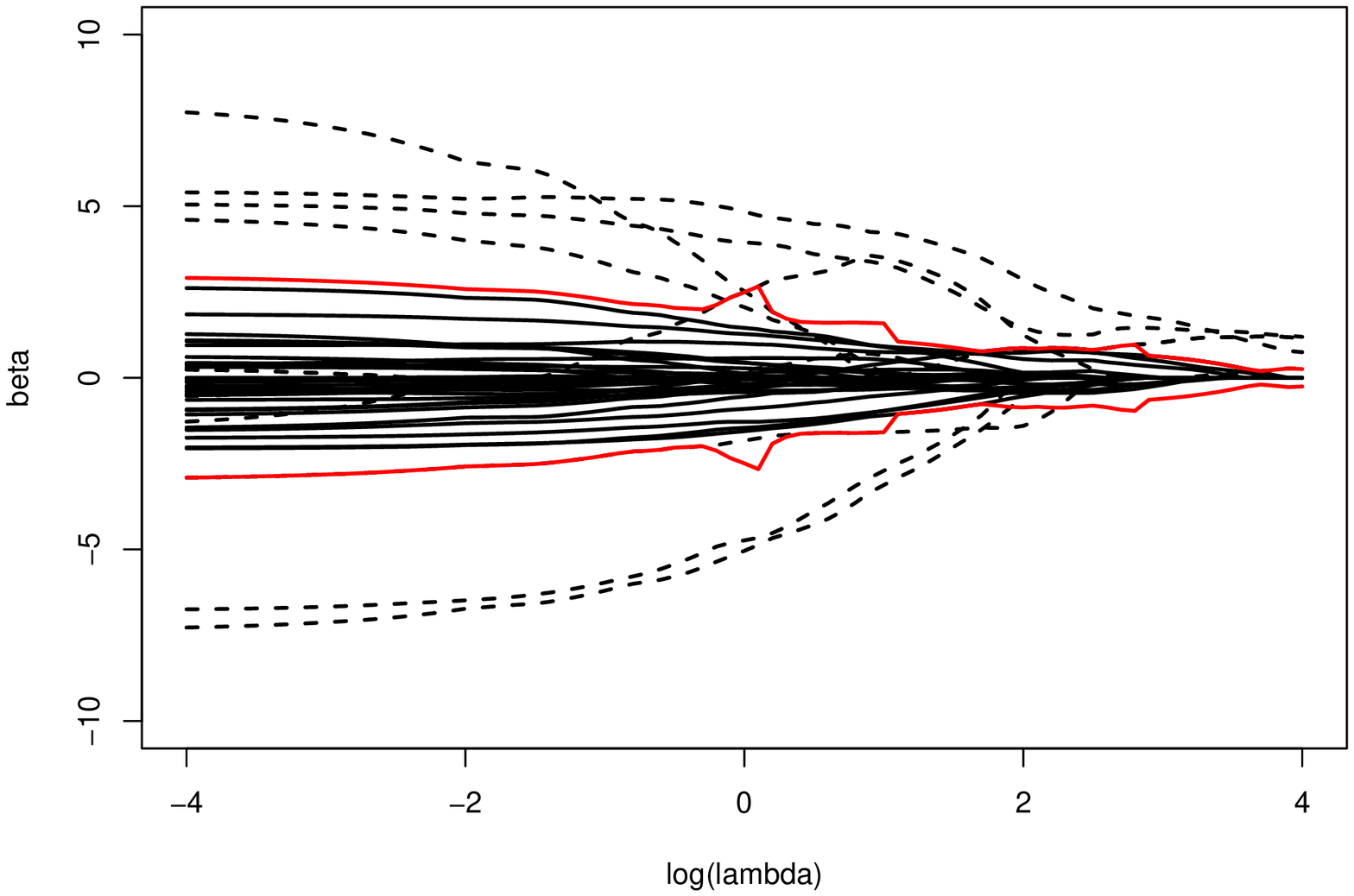}
		\includegraphics[width=0.49\columnwidth]{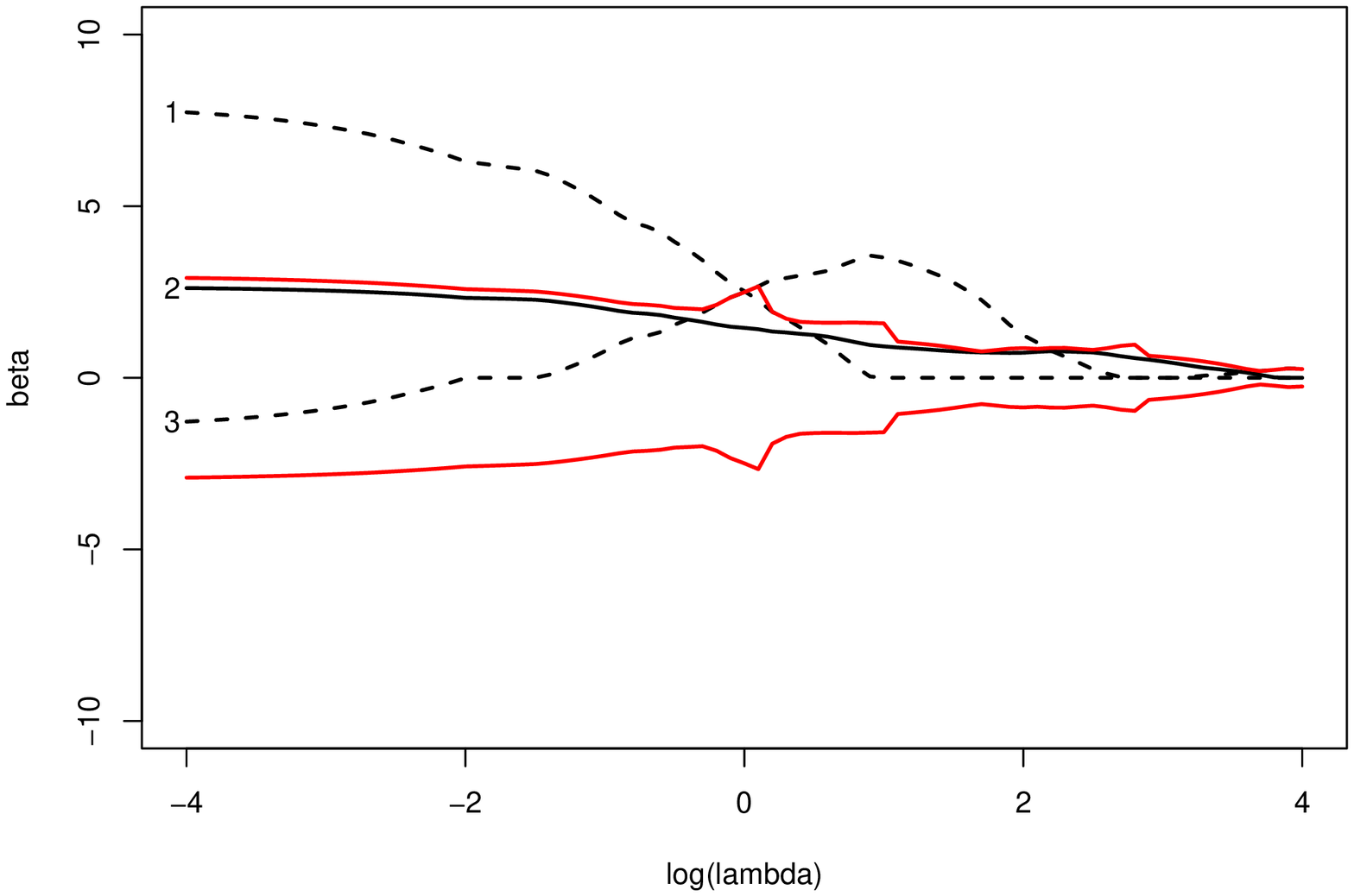}
		\caption{Left: Partitions of the lasso solution paths of the same simulated example. Right: Partitions of the lasso solution paths for the non-zero coefficients, 1 and 3, and the zero coefficient, 2.}
		\label{cluster}
	\end{center}
\end{figure}

The current restriction of utilizing just one tuning parameter can seriously reduce the accuracy of the feature selection in general, since solution paths like those in Figure~\ref{solution-path} happen quite frequently, and for any penalty we may employ. This is especially true when there are large correlations among the variables or the dimensions of the features are extremely high.  

Such deficiencies as described above motivate us to develop an approach that allows for collectively using evidence from the entire solution paths, as opposed to letting results from a single value of $\lambda$ dictate everything.  We achieve this by first dividing estimates from each vale of $\lambda$ into two clusters. Then we combine the cutoff points of all such clusters to form a boundary curve to partition the solution paths into  two regions, as shown by the red curves in Figure \ref{cluster}. We call the region inside the two red curves the zero region and the region outside the two curves the non-zero region. Finally, we choose all variables which have been identified as relevant variables for at least one value of $\lambda$ as the important features. We consider a feature to be unimportant if its solution path never goes out of the zero region.

 It can readily be observed in Figure \ref{cluster} that the SPSP procedure correctly selects $9$ of the $10$ relevant variables and drops all irrelevant ones, outperforming the results for any single value of $\lambda$ in terms of selection accuracy. 
 As a result, our approach can correctly identify the relevance of features like the three in Figure \ref{solution-path}, labeled 1, 2, and 3. {\color{red} We also applied the SPSP procedure on the solution paths of elastic-net for the same example, as seen in the right panel of Figure~\ref{af5}.}
 Another advantage of SPSP is that it does not require the coefficients of the unimportant variables  to shrink to zero;  therefore, it allows us to carry out feature selection with only a ridge regression.

The proposed SPSP procedure considers a feature important even if its solution path enters the non-zero region only once. The strategy may seem aggressive in identifying relevant variables. This is because we start the SPSP process rather conservatively, in the sense that we consider every variable ``unimportant" for the smallest value of $\lambda$, when we initiate the partitioning process.
The clustering at larger values of $\lambda$ depends on the results for the previous $\lambda$. Thus, the SPSP procedure combines a conservative starting point with an aggressive selection strategy to optimize the selection accuracy.

\begin{figure}
	\begin{center}
		\includegraphics[width=0.49\columnwidth]{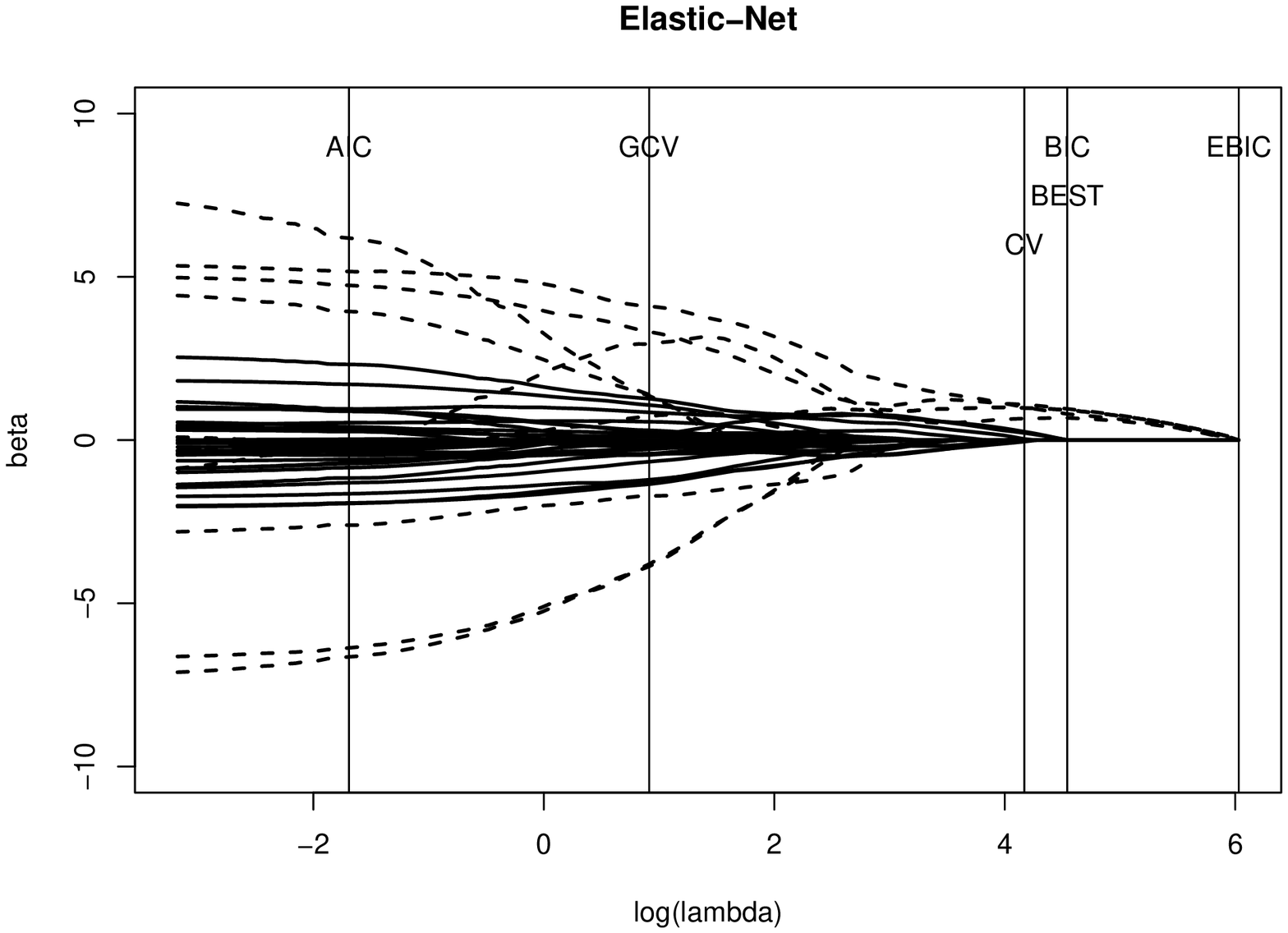}
		\includegraphics[width=0.49\columnwidth]{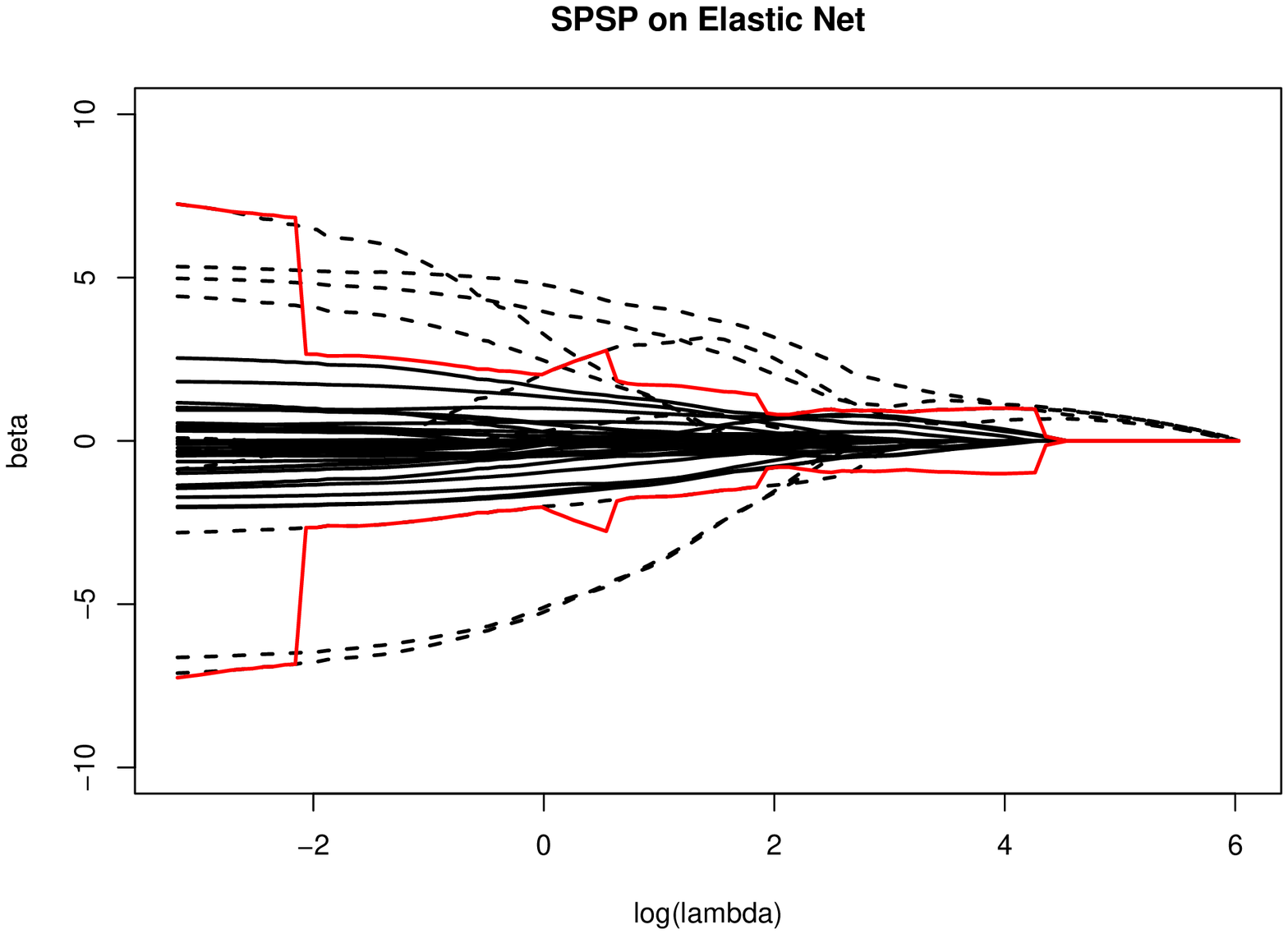}
		
		\caption{Left: The elastic net solution paths for the simulated example. The dashed lines are the paths of the non-zero coefficients, while the black lines are the paths of the zero coefficients 
			The vertical lines represent the tuning parameters selected by different criteria.
			Right: Partitions of the elastic net solution paths of the example using SPSP.
		}
		\label{af5}		
		
	\end{center}
\end{figure}

\section{Notations and Algorithm}
\subsection{Notations}
Consider the  penalized least squares problem in \eqref{penalty}.
We denote the index set for the non-zero coefficients as $S=\{j:\beta_j^\ast\neq 0\}$, with $s=|S|$, and the index set for the zero coefficients  as $S^c=\{j:\beta_j^\ast=0\}$. The goal of feature  selection is to correctly recover this sparsity pattern from the noisy observations in the model and estimate $S$. Once we specify the penalty function $J(\cdot)$ in \eqref{penalty}, we can compute the solution path on a grid for the tuning parameters $\lambda_1<\lambda_2<\cdots<\lambda_k$. In practice, the grid is usually 
equi-distant on the log scale \citep{buhlmann2011statistics,shen2012likelihood}. 
For each $\lambda_k$, we obtain a vector {\color{red}$\boldsymbol{\hat{\beta}^{(k)}}=(\hat \beta_{1}^{(k)}, \dots, \hat \beta_{p}^{(k)})^T$} of the penalized least squares estimators.
A variable is more likely to be identified as relevant if its estimator is farther away from $0$, regardless of the sign of the estimator. Therefore, we use the absolute values  {\color{red}$\boldsymbol{\hat {b}^{(k)}}=(\hat b_{1}^{(k)}, \dots, \hat b_{p}^{(k)})^T = (|\hat \beta_{1}^{(k)}|, \dots, |\hat \beta_{p}^{(k)}|)^T$} of the estimators.

In general, variables with a larger value of $\hat b_{j}^{(k)}$ are more likely to be important.  We are therefore interested in finding  a proper cutoff point $T_k=T(\lambda_k)$ such that the estimated relevant set $\hat S_k$ and irrelevant set $\hat S_k^c$ at $\lambda=\lambda_k$ are derived as $
\hat S_k =\{j : \hat b_{j}^{(k)} > T_k \}
$ and $
\hat S_k^c = \{ j:\hat b_{j}^{(k)}\leq T_k\}.
$
To obtain $(T_k,\hat S_k, \hat S_k^c )$ for each $\lambda_k$, we sort the absolute values {\color{red}$\hat b_{1}^{(k)}, \dots, \hat b_{p}^{(k)}$ in ascending order:
	$\hat b_{(1)}^{(k)}\leq \cdots \leq\hat b_{(p)}^{(k)},$ where $\hat b_{(j)}^{(k)}$ is the $j$th order statistics of $\hat b_{1}^{(k)}, \dots, \hat b_{p}^{(k)}$.}
Then we define the  adjacent distances between these ordered values as {\color{red}
	\begin{align}
	\label{distance}
	D_{j}^{(k)}= \hat b_{(j)}^{(k)} -\hat b_{(j-1)}^{(k)}, j=1, \dots, p. 
	\end{align}
} Note that $D_{1}^{(k)}$ is the adjacent distance between $\hat b_{(1)}^{(k)}$ and $0$, as we define $\hat b_{(0)}^{(k)}=0$ for convenience.
Letting $\hat s_k=|\hat S_k|$ be the number of variables in the estimated relevant set $\hat S_k$, there are $p-\hat s_k$ variables in the estimated irrelevant set $\hat S_k^c$. Hereafter,
we define the gap between $\hat S_k$ and $\hat S_k^c$ as the adjacent distance between $\hat b_{(p-\hat s_k)}^{(k)}$ and $\hat b_{(p-\hat s_k+1)}^{(k)}$, i.e.
{\color{red}$D(\hat S_k, \hat S_k^c)=D_{p-\hat s_k+1}^{(k)}=\hat b_{(p-\hat s_k+1)}^{(k)}-\hat b_{(p-\hat s_k)}^{(k)}.$} 
Further, let  $D_{\max} (\hat S_k) = \max \{D_{j}^{(k)}: j >p-\hat s_k+1 \}$ be the largest adjacent distance in $\hat S_k$, 
let $D_{\max} (\hat S_k^c) = \max \{D_{j}^{(k)}: j<p-\hat s_k+1 \}$  be  the largest adjacent distance in $\hat S_k^c$, and let  $$D_{\max2} (\hat S_k^c) = \max \{D_{j}^{(k)}: j<\tilde j, D_{\tilde j}^{(k)}=D_{\max} (\hat S_k^c)\} $$ be the largest adjacent distance under $D_{\max} (\hat S_k^c)$.


\subsection{Algorithm of SPSP}
{\color{red}

{To identify an important feature $j, j \in S$ correctly without introducing any spurious variables, the estimate $\hat \beta_{j}^{(k)}$ has to be larger than the estimates of zero coefficients at some $\lambda_k$. But unlike the currently used methods of using just one tuning parameter, we do not require all of estimate $\hat \beta_{j}^{(k)}$ ``concurrently" larger than all the estimates of zero coefficients at certain $\lambda_k$. Rather we just need $\hat \beta_{j}^{(k)}$ to be relatively larger at some $\lambda_k$ for variable $j$ to be selected. That is, even if  $\hat \beta_{j}^{(k)}$ and $\hat \beta_{j'}^{(k)}$, $j, j' \in S$ take larger values (compared to the estimates of zero coefficients) at different $\lambda_k$, we can still identify both $j$ and $j'$ without the cost of adding any false positive signals. The idea is to distinguish the larger estimates from the smaller ones at each $\lambda_k$ and consider those variables with the larger estimates as $\hat{S}_k.$}

{The major issue therefore is to find a proper way to separate the larger and smaller estimates at each $\lambda_k.$ It is inappropriate to use a constant threshold because for larger $\lambda_k$ all the estimates might be small. However, the order of  $\hat \beta_{j}^{(k)}, j \in S$ are still likely larger than that of estimates of zero coefficients. Following this intuition, we can set the threshold $T_k$ at $\lambda_k$ as $\{T_k: \hat \beta_{j}^{(k)} < C\hat \beta_{l}^{(k)},  \hat \beta_{j}^{(k)} > C\hat \beta_{s}^{(k)} \mbox{ for any } \hat \beta_{j}^{(k)}, \hat \beta_{l}^{(k)} >T_k, \hat \beta_{s}^{(k)} <T_k  \}$,  where $C$ is certain constant. This way,  we consider an estimate  $\hat \beta_{j}^{(k)}$ relatively larger if  $\hat \beta_{j}^{(k)} > C\hat \beta_{s}^{(k)} $ for any smaller estimates $\hat \beta_{s}^{(k)}$ and guarantee that all the larger estimates are of the same order. As we can imagine, such a threshold happens where the adjacent distance (3) is relatively large. Therefore, instead of comparing all the $\hat\beta_{j}^{(k)}$ pairwise to find $T_k$, we can equivalently try to find an adjacent distance that is large enough to separate $\hat S_k$ and $\hat S_k^c$.}

In principle,  $D(\hat S_k, \hat S_k^c)$, the gap between $\hat S_k$ and $\hat S_k^c$, should be sufficiently large to separate the irrelevant features from the important ones if 1. all $\hat \beta_{j}^{(k)}$ above the gap  $D(\hat S_k, \hat S_k^c)$ are of the same order, 2. all $\hat \beta_{j}^{(k)}$ above the gap $D(\hat S_k, \hat S_k^c)$ are of higher order than those below the gap.  It is not hard to verify that these are equivalent to the following statements in terms of adjacent distances: 1. the adjacent distances above the gap $D(\hat S_k, \hat S_k^c)$ are not of higher order than the gap; 2. the adjacent distances below the gap $D(\hat S_k, \hat S_k^c)$ are of lower order than the gap. 
For that reason, we consider $D(\hat S_k, \hat S_k^c)$ to be large enough if it meets the following two criteria:
\begin{eqnarray}
\label{rule1}
\frac{D_{\max} (\hat S_k)}{D(\hat S_k, \hat S_k^c)} \leq R,\\
\label{rule2}
\frac{D(\hat S_k, \hat S_k^c)}{D_{\max} (\hat S_k^c)} >R,
\end{eqnarray}
where $R$ is a  constant which we can estimate from the data. }

Based on this principle, we develop the SPSP algorithm as follows: \\
\noindent 
\begin{minipage}{\textwidth}
\noindent{\rule{\linewidth}{0.4pt}}

\noindent{Algorithm 1: Selection by Partitioning the Solution Paths (SPSP)}\label{SPSPA}
\begin{itemize}
	\item[1]  At $\lambda_1$, set the initial values to $T_1=\infty, \hat S_1=\emptyset,$ and $\hat S_1^c=\{1, \dots, p\}$, and estimate the constant $R$ as $R=D_{\max} (\hat S_1^c)/D_{\max2} (\hat S_1^c).$
	\item[2] At each $\lambda_k$, estimate $T_k, \hat S_k, \mbox{and } \hat S_k^c$ from  $T_{k-1}, \hat S_{k-1}, \hat S_{k-1}^c$, and {\color{red}$\boldsymbol{\hat {b}}^{(k)}$}.
	\begin{itemize}
		\item[2.1] Update $T_k=\max_{j \in \hat S_{k-1}^c} \hat {b}_{j}^{(k)}, \hat S_{k}=\{j:  \hat {b}_{j}^{(k)} > T_k\}, \hat S_{k}^c=\{j:  \hat {b}_{j}^{(k)} \leq T_k\}$;
		\item[2.2] Calculate $D_{1}^{(k)}, \dots, D_{p}^{(k)}$. Further, obtain $D_{\max} (\hat S_k^c), D_{\max 2} (\hat S_k^c)$,  and $D(\hat S_k, \hat S_k^c)$.
		\item [2.3] If $D(\hat S_k, \hat S_k^c) \leq R\times D_{\max} (\hat S_k^c) $   and $D_{\max} (\hat S_k^c) > R \times  D_{\max 2} (\hat S_k^c)$, $D_{\max} (\hat S_k^c)$ is the new gap between $S_k$ and $S_k^C$, i.e. $D(\hat S_k, \hat S_k^c)= D_{\max} (\hat S_k^c).$ Therefore we also update {\color{red}$$ T_k=\hat b_{(\tilde j-1)}^{(k)}, \hat S_{k}=\{j:  \hat {b}_{j}^{(k)} > T_k\}, \hat S_{k}^c=\{j:  \hat {b}_{j}^{(k)} \leq T_k\},$$ where $\tilde j$ is the location of $D_{\max} (\hat S_k^c)$  i.e.  $D_{\tilde j}^{(k)}=D_{\max} (\hat S_k^c).$ } \\
			{\color{red}Otherwise} $T_k, \hat S_k, \hat S_k^c$ remain unchanged as in Step 2.1.
	\end{itemize}
	\item[3] Make $k=k+1$ and repeat Step 2 until $k=K$.
	\item[4] Identify the union of all $\hat S_k$ as the index set for our selected relevant variables, i.e. $\hat S= \bigcup\limits_{k=1}^K \hat S_k. $
\end{itemize}
\noindent{\rule{\linewidth}{0.4pt}}
\end{minipage}

\begin{remark}
	Principles \eqref{rule1} and  \eqref{rule2} are implemented in Step 2.3. They are equivalent to the claim that the estimators of the non-zero coefficients should be of the same order and they should have a higher order than the estimators of the zero-coefficients.
	Instead of comparing every pair of the estimators, we just use adjacent distances for simplicity of computation. Therefore, finding the proper $T_k$ now transforms to finding an adjacent distance which is large
	enough---i.e. which satisfies \eqref{rule1} and  \eqref{rule2}---to be the gap between the estimators for the zero coefficients
	and those for the non-zero ones. The constant $R$, which we obtain from the data in Step 1, is used to decide whether $D_{\max} (\hat S_k^c)$ is large enough to replace the current $D(\hat S_k, \hat S_k^c)$. As a matter of fact, our final selection results are not sensitive to the value of $R$, as evidenced by the simulation study conducted in Appendix \ref{A2}.  
\end{remark}

%


\begin{remark}
	In Step 2 we find, for each $\lambda_{k}$, the cutoff point $T_k$, the location of the gap which distinguishes the relevant and irrelevant variables, based on the results from $\lambda_{k-1}$. This not only simplifies the computation process, but also makes the boundary line $T_k=T(\lambda_k)$ smoother to avoid unstable selection results. Specifically, in Step 2.1, we first use the largest estimated coefficients among those identified as ``zero coefficients'' for $\lambda_{k-1}$ as the current boundary. This takes care of the case where some coefficients in $\hat S_{k-1}$ become small and enter into the zero region at $\lambda_k$. In Step 2.2 and Step 2.3, we decide whether any adjacent distances within $\hat S_k^c$ are large enough to be considered as the new gap between the zero and non-zero coefficients. This handles the scenario where there are too few variables in $\hat S_k$, so that a ``large'' gap still exists.
\end{remark}

\begin{remark} 
	When the algorithm starts, with $\lambda_1$, we use the initial values set in Step 1, with all variables considered ``irrelevant''. In other words, we are conservative about identifying the relevant variables at the start of the procedure. This is because we implement an aggressive selection strategy in Step 4 by using the union of all $\hat S_k$s as our estimator $\hat S$ for the set of the relevant variables. This effectively balances out the very conservative initial setup to achieve the best selection accuracy.  Another reason for the conservative initial setup is that the estimations at the small $\lambda_k$ parameters are usually unstable, as there might be too many non-zeros in $\boldsymbol{\hat{\beta}_k}$  since the shrinkage effects of the penalty are minimal. As a result, even the estimation of a zero coefficient can possibly be rather large. 
\end{remark}


Once we identify the index set of all the relevant variables  $\hat S$, we estimate the regression parameters $\hat{\beta}_{\hat{S}}$ in a model that only includes the features that have been selected,
$
\mathbf{y}=\mathbf{X}_{\hat{S}}\boldsymbol{\beta}_{\hat{S}}+\bep,
$
where $\mathbf{X}_{\hat{S}}=(\mathbf{x_j})_{j\in\hat S}$ and $\boldsymbol{\beta}_{\hat{S}}=(\beta_j)_{j\in\hat S}^T$.
In most cases, the number of features in $\hat S$ is smaller than the sample size, so we just use the least squares estimator as $\hat{\beta}_{\hat{S}}$. If the number of selected variables is larger than the sample size, we can use a ridge regression with a small shrinkage factor.

One unique advantage of this procedure is that it can be applied not only to penalties like lasso, adaptive lasso, SCAD, and MCP; it can also be  applied to penalties which do not produce the sparse solutions, such as the $l_2$ penalty. Therefore, the procedure can greatly reduce computational complexity for feature selection problems, since strictly convex penalties like a ridge are easier to solve and their estimators have an explicit form.

In addition, the SPSP algorithm can be easily extended to handle  selection problems in a wide range of models, including graphical models, generalized linear models, and Cox's proportional hazards models. The penalized likelihood approach, which obtains a sparse estimate by solving an objective function consisting of a likelihood and a penalty function, is usually applied in these models. Consequently, we can apply the SPSP algorithm to penalized likelihood estimators in a similar fashion. A simulation example using SPSP on Gaussian graphical models is provided in Appendix \ref{A3}.

\section{Consistent Feature Selection}
\label{sec:verify}

In this section, we discuss the theoretical  properties of  the SPSP procedure with lasso for feature selection in a modern high-dimensional regime where $p>>n.$ Here we limit our efforts to linear regression, although the SPSP procedure is generally applicable to other selection problems as well. The technical proofs of all the lemmas and theorems in this section are available in Appendix \ref{A1}.

``Consistent variable selection'' for a procedure refers to the following property of its estimator $\hat S$: 
$$
P(\hat S= S) \rightarrow 1 \mbox{ as } n \rightarrow \infty.
$$
In most of the existing literature, it is only possible to achieve feature selection consistency if the tuning parameter is restricted to a specific interval. Moreover, the widths of such intervals are usually so small that they converge to 0---see, for example, \cite{fan2001variable}, \cite{fan2004nonconcave}, \cite{zhao2006model}, and \cite{zou2006adaptive}. 
It is also  commonly recognized that under the high-dimensional setting, $p>>n$, the Gram matrix $\hat\Sigma=\frac{1}{n}\mathbf{X}^T\mathbf{X}$ is degenerate, which raises many difficulties in controlling the values of the lasso estimator. Therefore, some conditions on the design matrix are always required to establish the consistency of the feature selection. The most typical condition is probably the following irrepresentable condition \citep{meinshausen2006high, yuan2007model, zhao2006model, zou2006adaptive}: $$\left|\frac{1}{n}\mathbf{X}_{S^c}^T\mathbf{X}_{S}\left(\frac{1}{n}\mathbf{X}_{S}^T\mathbf{X}_{S}\right)^{-1}\right|\leq 1-\eta,$$
where $\eta$ is a positive constant.
\cite{zhao2006model} showed that this condition is sufficient and almost necessary for lasso to be selection consistent. However, the condition is restrictive and difficult to verify in practice. Here, we will first show that the SPSP procedure is selection consistent under either a much weaker compatibility condition \citep{buhlmann2011statistics} or the restricted eigenvalue condition \citep{bickel2009simultaneous} if we can bound the tuning parameter to an interval of constant width, rather than one that is converging to zero. We further show that under a weak identifiability condition, the SPSP procedure achieves selection consistency for almost all values of the tuning parameter, i.e. the entire solution path. The weak identifiability condition is still weaker than the irrepresentable condition. 
%

We first introduce the following compatibility condition:
\begin{assumption}{\textbf{Compatibility Condition} \citep{buhlmann2011statistics, van2007deterministic}.}
	For some constant $\phi>0$ and for any vector $\boldsymbol{\zeta}$ satisfying $||\boldsymbol{\zeta}_{S^c}||_1\leq 3||\boldsymbol{\zeta}_{S}||_1$, the following compatibility condition holds:
	\bee
	||\boldsymbol{\zeta}_{S}||_1^2\leq \left(\boldsymbol{\zeta^T}\hat\Sigma\boldsymbol{\zeta}\right)s/\phi^2,
	\eee
	where $s=|S|$ is the dimension of $\bbeta_S$.
\end{assumption}

The  compatibility condition is based on the fact that the bias of the lasso estimator $\boldsymbol{\zeta}=\hat{\boldsymbol{\beta}}_L-\boldsymbol{\beta}^\ast$ satisfies $||\boldsymbol{\zeta}_{S^c}||_1\leq 3||\boldsymbol{\zeta}_{S}||_1$ with a probability close to $1$ \citep{bickel2009simultaneous, buhlmann2011statistics}.
Hence we can restrict ourselves to such vectors in the condition.
Several similar assumptions have also been proposed to establish the consistency property of the lasso, such as the restricted eigenvalue condition \citep{bickel2009simultaneous}, the restricted isometry condition \citep{candes2005decoding}, and the coherence condition \citep{bunea2007sparsity}. The relations among these conditions can be found in \citet{buhlmann2011statistics}. Under the compatibility condition, we can bound both the bias and the prediction error of the lasso:

\begin{lemma}\label{cc}
	\citep{buhlmann2011statistics} Suppose that the compatibility condition holds, and let $\lambda_0=2\sigma\sqrt{\frac{t^2+2\log p}{n}}$ for any $t>0$. Then for $\lambda\geq 2\lambda_0$, we have
	\bee
	\frac{1}{n}||\mathbf{X}(\boldsymbol{\hat\beta}-\boldsymbol{\beta}^\ast)||_2^2+
	\lambda||\boldsymbol{\hat\beta}-\boldsymbol{\beta}^\ast||_1\leq \frac{4\lambda^2s}{\phi^2}
	\eee
	with a probability of at least $1-2e^{-t^2/2}$.
\end{lemma}
%

This lemma implies the bound for the prediction error and
and the following bound for the $l_1$-error of the lasso estimator:
$$||\boldsymbol{\hat\beta}-\boldsymbol{\beta}^\ast||_1\leq \frac{4\lambda s}{\phi^2}.$$ The compatibility condition required to bound the above errors is substantially weaker than the irrepresentable condition, which is necessary for achieving consistent variable selection under the currently used framework of choosing a single $\lambda$. \cite{buhlmann2011statistics} showed that the irrepresentable condition actually implies the compatibility condition.

With the proposed SPSP procedure, we can accomplish consistent variable selection without the irrepresentable condition. This is because at each $\lambda$, we cluster the lasso estimators into two groups rather than labeling all the variables with non-zero coefficient estimates as important features. Consequently, we only need to bound the bias of the lasso estimators rather than shrink some coefficients to zeros. In what follows, we will first introduce some necessary notation and then present {\color{red}Theorem \ref{compact}}. The theorem shows that when $\lambda$ is not too large, the SPSP procedure identifies the true relevant set $S$ with a probability close to 1 with only the compatibility condition.

Let $\delta_\lambda=\frac{4\lambda s}{\phi^2}$ and $\delta_0=\delta_{\lambda_0},$ where $\lambda_0$ is defined as in Lemma 1.
We first sort the absolute values of the true non-zero coefficients {\color{red} $\{b^*_{j}=|\beta^*_j|, j\in S\}$} in ascending order to get {\color{red} $b^*_{(1)}, \dots, b^*_{(s)}$}, and we define the true adjacent distances as
{\color{red} $D_0=0, D_1=b^*_{(1)}, D_2=  b^*_{(2)} -b^*_{(1)}, \dots, D_s=b^*_{(s)} -b^*_{(s-1)}.$}
Let $C_0=\sqrt{\frac{D_{\max}+\delta_0}{\delta_0}}-1$, $D_{\max}=\max_{1 \leq i \leq p }\{D_i\}$,
and
$$
C=\frac{D_{\max}}{\min\{b^*_i: b^*_i>(2+C_0)\delta_0\}}.
$$
Moreover, let $$ C_{under}^{i}= \frac{D_{i}}{\max\{D_{i'}: i' < i\}}$$ for $i=2,\dots,s$ and $C_{under}^1=\infty,$ and $$ C_{upper}^{i}= \frac{D_{i}}{\max\{D_{i'}: i' > i\}}$$
for $i=1,\dots,s-1$ and $C_{upper}^s=\infty.$
In addition, let $R=1+C$. {We also sort the absolute values of the estimators from lasso in ascending order to get {\color{red} $\hat b_{(1)}, \dots, \hat b_{(p)}$}, and define the adjacent distances of the lasso estimators as
	{\color{red} $\hat D_1=\hat b_{(1)}, \hat D_2=  \hat b_{(2)} -\hat b_{(1)}, \dots, \hat D_p=\hat b_{(p)} -\hat b_{(p-1)}.$}

	\begin{theorem}
		\label{compact}
		Let
		$i_\lambda=
		\min\{i:C_{under}^i \geq R, C_{upper}^i \geq \frac{1}{C}, D_i>(1-\frac{R}{C_{under}^i})^{-1}(1+R)\delta_\lambda\} 
		$, and {\color{red} $S_{\lambda}=\{j: b^*_j \geq b^*_{(i_\lambda)}\}$}. Under the compatibility condition, if $\lambda > 2\lambda_0$, the following inequalities hold for the lasso estimator  with a probability of at least $1-2e^{-t^2/2}$: 	
		%
		%
		%
		%
		%
		\begin{align*}
		\frac{ \hat D(S_\lambda,S_\lambda^c)}{ \hat D_{\max}(S^c_\lambda)}&> R,\\
		\frac{\hat D_{\max}(S_\lambda)}{ \hat D(S_\lambda,S_\lambda^c)}&\leq R,
		\end{align*}
		where  { $\hat D(S_\lambda,S_\lambda^c) = \hat D_{p -s +i_\lambda}$, $\hat D_{\max}(S^c_\lambda) = \max\{\hat D_j:1\leq j < p -s +i_\lambda\}$ and $\hat D_{\max}(S_\lambda)=\max\{\hat D_j:p -s +i_\lambda < j\leq p\}$.}
	\end{theorem}

Let $\hat{S}_{\lambda}$ denote the important features identified at $\lambda$. Theorem \ref{compact} implies that $
P(\hat{S}_\lambda=S_\lambda) > 1-2e^{-t^2}.$
When the tuning parameter $\lambda$ is bounded by $\min_{j \in S}|\beta^*_j|> (1+R)4\lambda s/\phi^2$, we recover $S$ exactly with a probability close to 1:
\[
P(\hat{S}_\lambda=S) >1-2e^{-t^2}.
\]
Further, we conclude that for larger values of $\lambda$, i.e. such that $\max_{j \in S}|\beta^*_j|>(1+R)4\lambda s/\phi^2 \geq \min_{j \in S}|\beta^*_j|$, there are no false positive signals in $\hat S_\lambda$,
$
P(\hat{S}_\lambda\subset S) >1-2e^{-t^2}.
$
The following theorem then follows immediately from the fact that our SPSP estimator is $\hat S=\cup_\lambda \hat S_\lambda$:

\begin{theorem}
	\label{small}
	Let
	$i_{2\lambda_0}=
	\min\{i:C_{under}^i \geq R, C_{upper}^i \geq \frac{1}{C}, D_i>(1-\frac{R}{C_{under}^i})^{-1}(1+R)2\delta_{\lambda_0}\} 
	$  and let {\color{red} $S_{2\lambda_0}=\{j: b^*_j\geq b^*_{(i_{2\lambda_0})}\}$}. Under the compatibility condition, the SPSP estimator $\hat S$ over $\lambda \in [2\lambda_0, \frac{\phi^2 D_{\max}}{4s(1+R)})$ recovers $S_{2\lambda_0}$ with a probability of at least $1-2e^{-t^2}$:
	\[
	P(\hat S= S_{2\lambda_0}) > 1-2e^{-t^2}.
	\]
	In particular,  when $\min_{j \in S}|\beta^*_j|>(1+R)2\delta_0$,
	\[
	P(\hat S= S) > 1-2e^{-t^2}.
	\]
\end{theorem}

{\color{red}Theorem \ref{small}} suggests that the proposed SPSP procedure is consistent for variable selection under only the compatibility condition. With Theorem \ref{small}, we require that the tuning parameter $\lambda$ is not larger than $\frac{\phi^2 D_{\max}}{4s(1+R)},$ the lower bound of which can be obtained with prior information. When no such information is available, we would need SPSP estimators over larger values of $\lambda$ to not select any spurious variables. This is easy to verify in practice, since we can simply examine whether any new variables enter into the relevant set for larger values of $\lambda$. However, in order to theoretically guarantee such behavior of the solution paths, we need an additional condition, which is still substantially weaker than the irrepresentable condition:

\begin{assumption}{\textbf{Identifiability Condition}}
	Let $\eta>0$ be some constant. For any $\bar{\bbeta}= (\bar\bbeta_{S}, \bar{\bbeta}_{S^C})$, the following identifiability condition holds:
	\be
	\label{iden}
	\|\bX\bbeta^\ast-\bX_{S}\bar\bbeta_{S}-\bX_{S^C}\bar\bbeta_{S^C}\|^2 \geq \min_{\|\bbeta_S\|_1 \leq \|\bar\bbeta_S\|_1+ (1-\eta)\|\bar{\bbeta}_{S^c}\|_1}\| \bX\bbeta^\ast-\bX_{S}\bbeta_{S}  \|^2.
	\ee
\end{assumption}

The identifiability condition indicates that with the true set of relevant variables,  we can approximate the noiseless response $\mathbf{X}\bbeta^\ast$ at least as well as with any other set of variables under almost the same $l_1$ constraint. It is not difficult to verify that the condition is weaker than the irrepresentable condition.

\begin{lemma}
	The irrepresentable condition implies the identifiability condition.
\end{lemma}

On the right hand side of \eqref{iden}, the coefficients of the irrelevant variables are set as $\bbeta_{S^C}=0.$ In fact, we can further weaken the identifiability condition  by allowing $\bbeta_{S^c}$ to be non-zero on the right-hand side of \eqref{iden}. Instead, we only require $\|\bbeta_{S^c}\|_1$ to be smaller than  $\|\bbeta_S\|_1$ up to a constant $k$. On top of that, we also relax the inequality \eqref{iden} by taking $\kappa\eta\|\hat{\bbeta}_{S^c}\|_1$  on the right side. As a result, we obtain the following weak identifiability condition:

\begin{assumption}{\textbf{Weak Identifiability Condition}}
	Let $\eta>0$ be some constant. For any $\bar{\bbeta}= (\bar\bbeta_{S}, \bar{\bbeta}_{S^C})$, then for $k=\frac{2}{2s+Rs(s+1)}$ and some $\kappa$ that satisfies 
	$$
	D_{\max} > \lambda_0 \frac{4s(1+R)}{\phi^2}\left\{\frac{Rs^2+(2+R)S+2}{\eta}-1+\kappa \right\},
	$$
	the following weak identifiability condition holds,
	\begin{align}
	\label{weak_iden}
	\|\mathbf{X}\bbeta^\ast-\bX_{S}\bar\bbeta_{S}-\bX_{S^C}\bar\bbeta_{S^C}\|^2 \geq \min_{\bbeta \in \Theta(\|\bar{\bbeta}_S\|_1, \|\bar{\bbeta}_{S^C}\|_1)} \| \bX\bbeta^\ast-\bX\bbeta  \|^2 - \kappa\eta\|\bar{\bbeta}_{S^c}\|_1,
	\end{align}
	where $\Theta(\|\bar{\bbeta}_S\|_1, \|\bar{\bbeta}_{S^C}\|_1)=\{\bbeta = (\bbeta_S, \bbeta_{S^C}): \|\bbeta\|_1 \leq \|\bar{\bbeta}_S\|_1+ (1-\eta)\|\bar{\bbeta}_{S^C}\|_1, \|\bbeta_{S^C}\|_1 \leq k \|\bbeta_{S}\|_1 \}. $ 
\end{assumption}
Henceforth, we refer to the preceding weak identifiability condition with constants $k$ and $\kappa$ as $WIC(k, \kappa).$  Apparently,  $WIC(0, 0)$ is simply the identifiability condition in \eqref{iden}, and $WIC(k, \kappa)$ always implies $WIC(k', \kappa')$ for $k'>k, \kappa' >\kappa$. Therefore Assumption 3 is always weaker than Assumption 2. The above assumption ensures that when $\lambda$ is large, the lasso estimates for the zero coefficients will not be much larger than those for the non-zero coefficients, so that we will not have any false positive signals from our SPSP procedure in $\hat S.$ We combine this consideration with Theorem \ref{compact} to obtain the following result for the entire solution paths under the compatibility condition and $WIC(k, \kappa)$:

\begin{theorem}
	\label{final}
	Suppose that under the compatibility condition and $WIC(k,\kappa)$ with $k=\frac{2}{2s+Rs(s+1)}$, $$
	D_{\max} > \lambda_0 \frac{4s(1+R)}{\phi^2}\left\{\frac{Rs^2+(2+R)S+2}{\eta}-1+\kappa \right\}.
	$$
	Then the SPSP procedure over $\lambda \in [2\lambda_0, \infty)$ identifies $S_{2\lambda_0}=\{j: |\beta^*_j|> (1+R)2\delta_0\}$ with a probability of at least $1-2e^{-t^2},$ i.e.
	$$
	P(\hat S=S_{2\lambda_0}) >1-2e^{-t^2}.
	$$
	
\end{theorem}

When the true values of the coefficients are of a higher order than $\sqrt{\frac{\log p} { n}}$, it follows immediately from {\color{red}Theorem \ref{final}} that the asymptotic probability of identifying the true relevant set $S$ is 1. Here we only need the tuning parameter $\lambda$ to be not too small: $\lambda>2\lambda_0$; unlike the existing literature, we do not require the tuning parameter $\lambda$ to be in a specific region. In fact, we do not need consistent variable selection for any value of $\lambda$. We only need to control the bias of the lasso estimators for smaller $\lambda$s and to control the $l_1$ norm of the lasso estimators of those zero coefficients for larger $\lambda$s; both of these results require weaker conditions than achieving selection consistency at certain values of $\lambda$. By combining these results, the SPSP procedure can accomplish feature selection consistency under substantially weaker conditions, without a proper choice of the tuning parameter.

{
	\color{red}
The theoretical results for other penalty functions can be developed following the paradigm as described above.  We can derive Theorem \ref{small} given the bias of the penalized estimators and the corresponding condition. For the non-convex penalties, the bias is smaller than lasso. For $l_2$ penalty, \cite{shaoEstimationHighdimensionalLinear2012} gives the bias and convergence rate for high-dimensional ridge regression. To obtain Theorem \ref{final}, we need a similar identifiability condition as those in \eqref{iden} and \eqref{weak_iden}. However, the constraint on the $l_1$ norm needs to be replaced by the constraint on $l_2$ norm for ridge regression, and by the corresponding function form for the non-convex penalties.

}

\section{Simulation Studies}
All of the following simulations are generated from the linear model \eqref{lm} with $x_{ij}\sim N(0,1),i=1,\dots,n,j=1,\dots,p$, and $\epsilon_i \sim N(0,\sigma^2)$. The details of the simulation setups are as follows:

\begin{enumerate}
	
	\item[(M1)]  (Moderate Correlation, $p>n$) \citep{tibshirani1996regression} Let $\beta_1^\ast=3,\beta_2^\ast=1.5,$ and $\beta_5^\ast=2$, and let the remaining coefficients equal zero. The correlation between $x_{j_1}$ and $x_{j_2}$ is $0.5^{|j_1-j_2|}$. We set $n=50$, $p=100$, and $\sigma =3$.
	
	\item[(M2)]  (Moderate Correlation, $p>>n$) The setup here is the same as (M1), except that $p=1000$. 
	
	
	\item[(M3)] (High Correlation, $p>n$) \citep{wang2011random} Let $\boldsymbol{\beta}^\ast=(3,3,-2,3,3,-2,0,\dots,0)$, so that the first $6$ coefficients are non-zero and the remaining $94$ coefficients equal zero.  The pairwise correlation between the first $3$ variables is $0.9$, the pairwise correlation between the second $3$ variables is also $0.9$, and the remaining $94$ variables are independent of each other. Furthermore, the first $3$ variables, the second $3$ variables, and the remaining $94$ variables are independent of each other. We set $n=50$ and $\sigma=3$.
	%
	
	\item[(M4)] (Misspecified Model, $p>n$) \citep{lv2014model} The true values of the coefficients are $\boldsymbol{\beta}^\ast=(1,-1.25,0.75,-0.95,1.5,0,\dots,0),$ where $p=100$, the first 5 coefficients are non-zero, and the remaining coefficients equal zero. All of the variables are independent of each other, i.e the rows of $\mathbf{X}$ are generated from $N(\mathbf{0}, \mathbf{I}_p)$. We set $n=50$ and $\sigma=1$. Here the response $\mathbf{y}$ is generated as
	$\mathbf{y}=\mathbf{X}\boldsymbol{\beta^\ast}+\mathbf{x}_{p+1}+\bep,$
	where $\mathbf{x}_{p+1}=\mathbf{x}_1\circ\mathbf{x}_2$ is an interaction term between $x_1$ and $x_2$. However, ${x}_{p+1}$ is not a candidate variable, and therefore the true model is not contained in the candidate models. 
\end{enumerate}
%

For all simulation examples, we conduct the SPSP procedure on the solution paths for lasso, adaptive lasso, SCAD, and MCP. One compelling advantage of the proposed SPSP procedure is that it can be applied for the $l_2$ penalty. Therefore, we also implement the SPSP algorithm with ridge regression for these examples. 
We compute the solution paths for convex penalties  with  the $\mathbf{R}$ package \textit{glmnet} and those for SCAD and MCP with the $\mathbf{R}$ package \textit{ncvreg}. All of the solution paths are computed over a grid of $K=100$ values of the tuning parameter $\lambda$. 


In Tables 2--5, we compare the performance of the proposed SPSP procedure with the $10$-fold cross-validation, generalized cross-validation, AIC, BIC, extended BIC and stability selection over $500$ replicates
for each setup. We record the following measures: FP, the number of false positives, FN, the number of false negatives, and ME, the Model Error = $(\hat{\boldsymbol{\beta}}-\boldsymbol{\beta}^\ast)^T\hat{\boldsymbol{\Omega}}(\hat{\boldsymbol{\beta}}-\boldsymbol{\beta}^\ast)/\sigma^2$, where $\hat{\boldsymbol{\Omega}}$ is the sample covariance of $\mathbf{X}$. We report the mean values of false positives and false negatives along with the median of the model error, because the distribution of model error  values is heavily skewed. We also report the standard errors for the above measures in parentheses.   

%
%


Table~\ref{Ex1} summarizes the results for M1, and we can observe that cross-validation, generalized cross-validation, AIC, and BIC on all the penalties  tend to select too many variables in the model, as evidenced by the extremely large false positive numbers. The performance of the extended BIC here is comparable to that of SPSP, since SPSP produces a smaller number false negatives and slightly larger number of false positives. Stability selection performs well for adaptive lasso, but yields large false negative values for the other three penalties, especially the two non-convex ones: it misses two of the three signals on average for MCP and SCAD. The SPSP procedure provides the smallest model error for all penalties except for adaptive lasso, where the model error for SPSP is slightly larger than that for stability selection.

%


The simulation results for M2 are summarized in  Table~\ref{Ex2}. We find that our SPSP approach has a more significant advantage over all of the other approaches when $p >>n$, compared to the results for M1. The SPSP procedure produces much smaller false positive values than cross-validation, generalized cross-validation, AIC, and BIC for lasso and adaptive lasso,  while its false negative values are close to those for the other approaches. On the other hand, the extended BIC and stability selection often identify 1 or no important signals, although their false positive values are slightly smaller than the false positive value for SPSP. Here cross-validation has a similar performance to SPSP for the non-convex penalties, but it performs rather poorly for lasso and adaptive lasso. Finally, the performance of SPSP with ridge regression is very close  to that of the other penalties, which proves that even the $l_2$ penalty can be effective in feature selection with the SPSP procedure.   
%
%

Table~\ref{Ex3} presents the simulation results for M3, where the true model is very sparse and all of the important variables are correlated.  We observe that compared to cross-validation, generalized cross-validation, AIC, and BIC, the SPSP procedure can dramatically improve the selection accuracy by selecting fewer irrelevant variables. We also observe that stability selection identifies fewer than one of the six signals with the non-convex penalties. Moreover, SPSP produces the smallest model error among all the approaches for all of the penalties, although in this case, the extended BIC also shows competitive performance.   

%

The results for M4 are shown in Table~\ref{Ex4}. We observe that when models are misspecified, the SPSP procedure enjoys a substantially better performance for all four penalties in both selection accuracy and model error compared to the other selection criteria, except that the performance of cross-validation is similar to SPSP for the two non-convex penalties.

{
	\color{red}
	For each model in the simulation studies, we further select one example to illustrate the advantage of the SPSP approach on lasso solution paths. We plot these solution paths and the partitions using SPSP in Figure \ref{af4}.  
}

In summary, the proposed SPSP approach provides the best or close-to-best performance for all of the penalties in all simulation examples. Generally speaking, cross-validation, generalized cross-validation, AIC, and BIC tend to select too many spurious variables, while the extended BIC and stability selection tend to ignore important features and in some cases even fail to identify any useful signals. Moreover, both stability selection and cross-validation cost significantly more computation time than SPSP. In our simulation studies, the computational cost of stability selection is 30--60 times that for SPSP.  Finally, the performance of SPSP with ridge regression is quite similar to the other penalties, also generally outperforming the other selection criteria. This establishes the promise of potential applications of selection using the $l_2$ penalty, a unique feature of SPSP. 
\begin{table}[h]
	\centering
		\centering
		\caption{Simulation results for Model 1 over 100 replicates.}
			
			\begin{tabular}{ccccccccc}
				\hline
				& & CV & GCV & AIC & BIC & EBIC & STAB & SPSP \\ 
				\hline
				Lasso & FP & 12.508 & 21.154 & 43.082 & 41.312 & 0.23 & 0.024 & 4.476 \\ 
				&  & (0.444) & (0.6) & (0.12) & (0.332) & (0.024) & (0.007) & (0.393) \\ 
				& FN & 0.042 & 0.064 & 0.088 & 0.092 & 0.848 & 1.166 & 0.37 \\ 
				&  & (0.009) & (0.012) & (0.013) & (0.014) & (0.047) & (0.031) & (0.027) \\ 
				& ME & 0.271 & 0.442 & 0.806 & 0.798 & 0.456 & 0.445 & 0.253 \\ 
				&  & (0.01) & (0.016) & (0.01) & (0.012) & (0.034) & (0.015) & (0.018) \\ 
				\hline
				adaLasso & FP & 9.084 & 9.876 & 18.402 & 3.216 & 0.212 & 0.504 & 1.558 \\ 
				&  & (0.302) & (0.294) & (0.32) & (0.209) & (0.022) & (0.032) & (0.161) \\ 
				& FN & 0.036 & 0.042 & 0.032 & 0.104 & 0.628 & 0.382 & 0.446 \\ 
				&  & (0.009) & (0.009) & (0.008) & (0.015) & (0.035) & (0.024) & (0.03) \\ 
				& ME & 0.22 & 0.269 & 0.414 & 0.17 & 0.225 & 0.159 & 0.209 \\ 
				&  & (0.01) & (0.01) & (0.01) & (0.009) & (0.018) & (0.011) & (0.014) \\ 
				\hline
				SCAD & FP & 2.388 & 18.496 & 19.476 & 17.736 & 1.414 & 0 & 2.22 \\ 
				&  & (0.103) & (0.132) & (0.125) & (0.187) & (0.195) & (0) & (0.186) \\ 
				& FN & 0.442 & 0.408 & 0.4 & 0.42 & 0.838 & 1.722 & 0.492 \\ 
				&  & (0.024) & (0.024) & (0.024) & (0.025) & (0.037) & (0.026) & (0.03) \\ 
				& ME & 0.28 & 0.8 & 0.816 & 0.787 & 0.35 & 0.66 & 0.251 \\ 
				&  & (0.011) & (0.013) & (0.013) & (0.014) & (0.028) & (0.026) & (0.015) \\ 	\hline
				MCP & FP & 0.928 & 17.902 & 18.192 & 17.692 & 2.972 & 0 & 2.38 \\ 
				
				&  & (0.06) & (0.124) & (0.121) & (0.134) & (0.286) & (0) & (0.167) \\ 
				& FN & 0.678 & 0.428 & 0.426 & 0.428 & 0.82 & 2.302 & 0.574 \\ 
				&  & (0.026) & (0.024) & (0.024) & (0.024) & (0.03) & (0.028) & (0.029) \\ 
				& ME & 0.268 & 0.873 & 0.873 & 0.873 & 0.265 & 0.773 & 0.266 \\ 
				&  & (0.01) & (0.012) & (0.012) & (0.012) & (0.019) & (0.049) & (0.015) \\ 
				\hline
				Ridge & FP &  &  & &  &  & &	3.282 \\ 
				& &  &  & &  &  & & (0.621) \\ 
				& FN &  &  & &  &  & &0.932 \\ 
				&  & &  &  & &  &  & (0.038) \\ 
				& ME &  &  & &  &  & &	0.472 \\ 
				&  &  &  & &  &  & & 	(0.02) \\ 
				\hline
			\end{tabular}
			
			\label{Ex1}
\end{table}

\begin{table}
		\centering
		\caption{Simulation results for Model 2 over 100 replicates.}
			\begin{tabular}{ccccccccc}
				\hline
				& & CV & GCV & AIC & BIC & EBIC & STAB & SPSP \\ 
				\hline
				Lasso & FP & 23.572 & 36.526 & 48.302 & 48.178 & 0.02 & 0.002 & 2.126 \\ 
				&  & (0.632) & (0.601) & (0.122) & (0.12) & (0.006) & (0.002) & (0.212) \\ 
				& FN & 0.142 & 0.13 & 0.122 & 0.122 & 2.25 & 2.3 & 0.712 \\ 
				&  & (0.017) & (0.016) & (0.015) & (0.015) & (0.043) & (0.025) & (0.033) \\ 
				& ME & 0.51 & 0.765 & 0.895 & 0.895 & 1.877 & 0.776 & 0.428 \\ 
				&  & (0.014) & (0.016) & (0.011) & (0.011) & (0.044) & (0.044) & (0.017) \\  	\hline
				adaLasso & FP & 31.228 & 28.382 & 38.772 & 36.354 & 0.032 & 0.208 & 3.138 \\ 
				&  & (0.431) & (0.388) & (0.159) & (0.372) & (0.008) & (0.021) & (0.305) \\ 
				& FN & 0.126 & 0.138 & 0.13 & 0.132 & 1.602 & 0.956 & 0.672 \\ 
				&  & (0.016) & (0.016) & (0.016) & (0.016) & (0.045) & (0.033) & (0.033) \\ 
				& ME & 0.671 & 0.666 & 0.789 & 0.781 & 0.75 & 0.383 & 0.449 \\  
				&  & (0.014) & (0.014) & (0.012) & (0.014) & (0.04) & (0.016) & (0.017) \\ 	\hline
				SCAD & FP & 6.086 & 22.248 & 23.842 & 22.302 & 0.084 & 0 & 6.53 \\ 
				&  & (0.213) & (0.135) & (0.154) & (0.143) & (0.013) & (0) & (0.399) \\ 
				& FN & 0.602 & 1.006 & 1.01 & 1.006 & 2.112 & 2.222 & 0.564 \\ 
				&  & (0.027) & (0.029) & (0.029) & (0.029) & (0.045) & (0.024) & (0.031) \\ 
				& ME & 0.389 & 0.887 & 0.895 & 0.887 & 1.469 & 0.751 & 0.559 \\ 
				&  & (0.013) & (0.012) & (0.012) & (0.012) & (0.049) & (0.041) & (0.021) \\ 	\hline
				MCP & FP & 1.384 & 15.378 & 15.574 & 15.266 & 3.786 & 0 & 3.408 \\ 
				&  & (0.079) & (0.072) & (0.074) & (0.075) & (0.269) & (0) & (0.214) \\ 
				& FN & 0.986 & 0.928 & 0.928 & 0.928 & 1.338 & 2.67 & 0.922 \\ 
				&  & (0.029) & (0.028) & (0.028) & (0.028) & (0.033) & (0.021) & (0.032) \\ 
				& ME & 0.41 & 0.932 & 0.932 & 0.932 & 0.705 & 2.143 & 0.575 \\ 
				&  & (0.013) & (0.012) & (0.012) & (0.012) & (0.027) & (0.046) & (0.019) \\ 	\hline
				Ridge & FP &  &  & &  &  & &	3.134 \\ 
				& &  &  & &  &  & & 	(1.79) \\ 
				& FN &  &  & &  &  & &	0.944 \\ 
				&  & &  &  & &  &  & 	(0.031) \\ 
				& ME &  &  & &  &  & &	0.472 \\ 
				&  &  &  & &  &  & & 	(0.013) \\ 
				\hline
			\end{tabular}
			\label{Ex2}
\end{table}
\begin{table}[h]
		\centering
		\caption{Simulation results for Model 3 over 100 replicates.}
			\begin{tabular}{ccccccccc}
				\hline
				&  & CV & GCV & AIC & BIC & EBIC & STAB & SPSP \\ 
				\hline
				Lasso & FP & 12.846 & 22.674 & 42.84 & 42.08 & 0.354 & 0.02 & 3.222 \\ 
				&  & (0.459) & (0.591) & (0.117) & (0.207) & (0.111) & (0.006) & (0.311) \\ 
				& FN & 2.27 & 2.268 & 2.208 & 2.224 & 3.124 & 4.754 & 2.6 \\ 
				&  & (0.026) & (0.028) & (0.032) & (0.031) & (0.056) & (0.035) & (0.038) \\ 
				& ME & 0.395 & 0.543 & 0.827 & 0.823 & 0.623 & 1.844 & 0.306 \\ 
				&  & (0.01) & (0.013) & (0.011) & (0.011) & (0.049) & (0.068) & (0.016) \\ 	\hline
				adaLasso & FP & 10.144 & 11.572 & 21.172 & 4.026 & 0.146 & 0.554 & 2.06 \\ 
				&  & (0.32) & (0.325) & (0.268) & (0.285) & (0.019) & (0.033) & (0.183) \\ 
				& FN & 1.83 & 1.794 & 1.218 & 2.274 & 2.92 & 3.536 & 2.424 \\ 
				&  & (0.036) & (0.037) & (0.039) & (0.034) & (0.038) & (0.032) & (0.043) \\ 
				& ME & 0.339 & 0.37 & 0.507 & 0.29 & 0.365 & 0.337 & 0.292 \\ 
				&  & (0.009) & (0.01) & (0.01) & (0.009) & (0.021) & (0.026) & (0.013) \\ 	\hline
				SCAD & FP & 1.308 & 19.776 & 20.894 & 18.28 & 0.594 & 0 & 0.696 \\ 
				&  & (0.092) & (0.147) & (0.139) & (0.265) & (0.11) & (0) & (0.126) \\ 
				& FN & 3.86 & 3.96 & 3.958 & 3.962 & 4.01 & 5.378 & 3.002 \\ 
				&  & (0.017) & (0.009) & (0.009) & (0.009) & (0.012) & (0.03) & (0.035) \\ 
				& ME & 0.39 & 0.813 & 0.827 & 0.804 & 0.37 & 2.207 & 0.267 \\ 
				&  & (0.009) & (0.012) & (0.012) & (0.014) & (0.019) & (0.067) & (0.009) \\ 	\hline
				MCP & FP & 0.574 & 18.718 & 19.002 & 18.506 & 2.438 & 0 & 0.984 \\ 
				&  & (0.058) & (0.127) & (0.125) & (0.135) & (0.29) & (0) & (0.131) \\ 
				& FN & 4 & 3.984 & 3.984 & 3.984 & 3.998 & 5.808 & 3.986 \\ 
				&  & (0) & (0.006) & (0.006) & (0.006) & (0.002) & (0.02) & (0.005) \\ 
				& ME & 0.375 & 0.906 & 0.906 & 0.906 & 0.363 & 3.572 & 0.353 \\ 
				&  & (0.007) & (0.012) & (0.012) & (0.012) & (0.014) & (0.052) & (0.01) \\ 
				\hline
				Ridge & FP &  &  & &  &  & &		12.79 \\ 
				& &  &  & &  &  & & 	(1.248) \\ 
				& FN &  &  & &  &  & &		1.36 \\ 
				&  & &  &  & &  &  & 		(0.073) \\ 
				& ME &  &  & &  &  & &		0.34 \\ 
				&  &  &  & &  &  & & 		(0.039) \\ 
				\hline
			\end{tabular}
			\label{Ex3}
\end{table}

\begin{table}
		\centering
		\caption{Simulation results for Model 4 over 100 replicates.}
			\begin{tabular}{ccccccccc}
				\hline
				&  & CV & GCV & AIC & BIC & EBIC & STAB & SPSP \\ 
				\hline
				Lasso & FP & 17.438 & 22.624 & 42.604 & 40.322 & 2.2 & 0.002 & 4.312 \\ 
				&  & (0.416) & (0.512) & (0.118) & (0.352) & (0.396) & (0.002) & (0.365) \\ 
				& FN & 0.142 & 0.124 & 0.098 & 0.13 & 3.746 & 3.112 & 1 \\ 
				&  & (0.018) & (0.017) & (0.014) & (0.02) & (0.084) & (0.044) & (0.052) \\ 
				& ME & 0.926 & 1.145 & 1.742 & 1.725 & 5.377 & 2.797 & 1.146 \\ 
				&  & (0.024) & (0.028) & (0.022) & (0.026) & (0.125) & (0.067) & (0.048) \\  		\hline
				adaLasso & FP & 13.15 & 13.94 & 25.526 & 9.976 & 0.362 & 0.496 & 2.91 \\ 
				&  & (0.316) & (0.327) & (0.28) & (0.43) & (0.038) & (0.032) & (0.242) \\ 
				& FN & 0.25 & 0.236 & 0.136 & 0.36 & 2.534 & 1.098 & 1.048 \\ 
				&  & (0.023) & (0.023) & (0.016) & (0.029) & (0.083) & (0.041) & (0.049) \\ 
				& ME & 0.872 & 0.936 & 1.213 & 0.922 & 2.081 & 0.906 & 1.059 \\ 
				&  & (0.023) & (0.023) & (0.022) & (0.027) & (0.113) & (0.038) & (0.045) \\ 		\hline
				SCAD & FP & 3.914 & 19.242 & 20.484 & 19.006 & 5.114 & 0 & 4.618 \\ 
				&  & (0.109) & (0.13) & (0.128) & (0.159) & (0.35) & (0) & (0.272) \\ 
				& FN & 0.314 & 0.324 & 0.324 & 0.326 & 1.962 & 4.094 & 0.574 \\ 
				&  & (0.028) & (0.028) & (0.028) & (0.028) & (0.096) & (0.034) & (0.04) \\ 
				& ME & 0.918 & 1.707 & 1.741 & 1.707 & 1.902 & 4.035 & 0.996 \\ 
				&  & (0.034) & (0.025) & (0.025) & (0.026) & (0.128) & (0.079) & (0.039) \\ 		\hline
				MCP & FP & 1.602 & 18.282 & 18.664 & 18.204 & 11.228 & 0 & 3.96 \\ 
				&  & (0.079) & (0.118) & (0.113) & (0.123) & (0.394) & (0) & (0.223) \\ 
				& FN & 0.562 & 0.33 & 0.326 & 0.33 & 0.646 & 4.63 & 0.548 \\ 
				&  & (0.038) & (0.027) & (0.027) & (0.027) & (0.047) & (0.025) & (0.04) \\ 
				& ME & 0.829 & 1.831 & 1.836 & 1.831 & 1.711 & 5.516 & 1.029 \\ 
				&  & (0.037) & (0.024) & (0.024) & (0.024) & (0.047) & (0.084) & (0.041) \\

				\hline
				Ridge & FP &  &  & &  &  & &			3.528 \\ 
				& &  &  & &  &  & & 	(0.628) \\ 
				& FN &  &  & &  &  & &		2.336 \\ 
				&  & &  &  & &  &  & 		(0.059) \\ 
				& ME &  &  & &  &  & &		2.093 \\ 
				&  &  &  & &  &  & & 		(0.073) \\ 
				\hline
			\end{tabular}
			\label{Ex4}
\end{table}

\begin{figure}
	\begin{center}
		\includegraphics[width=0.49\columnwidth]{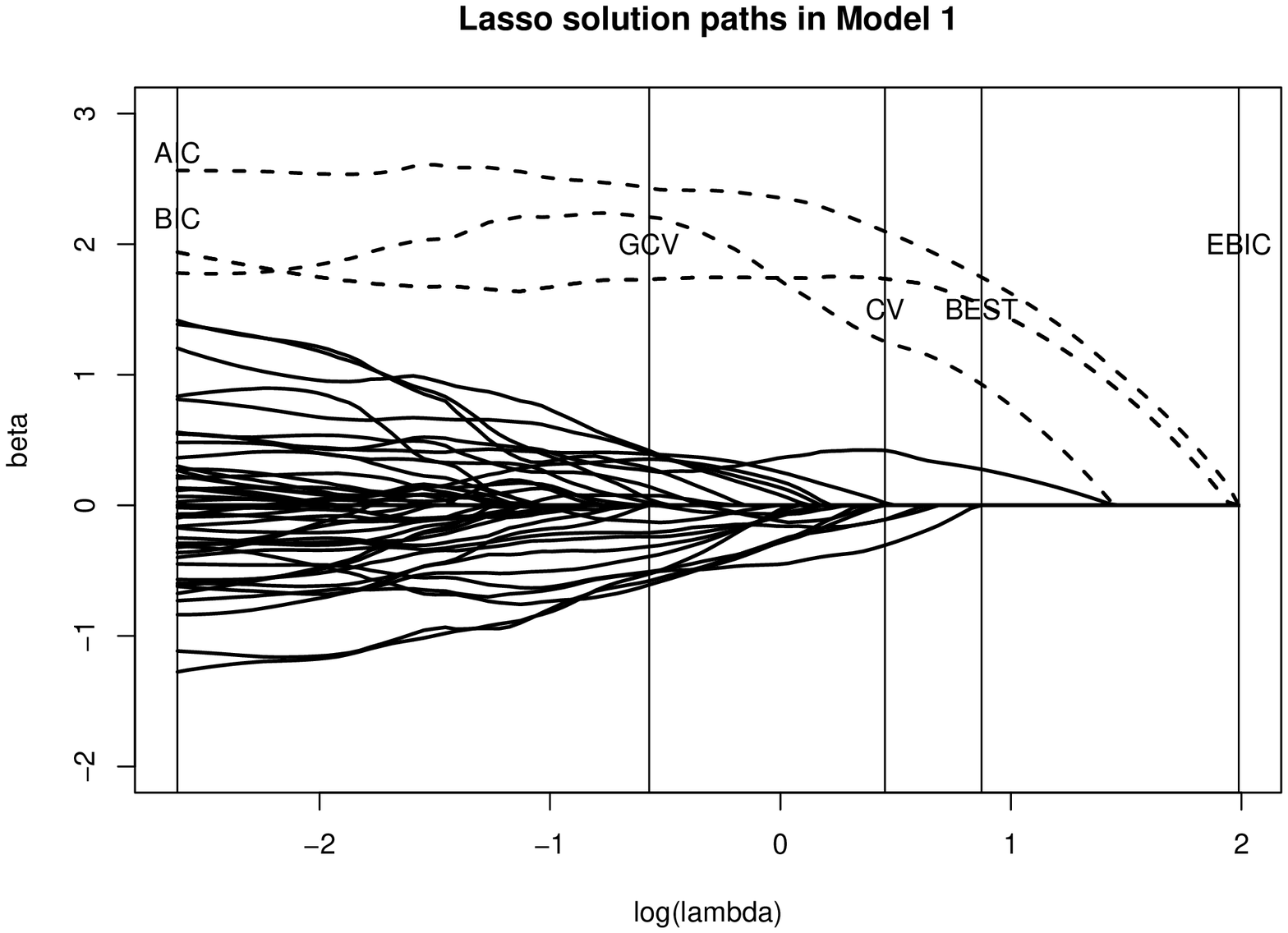}
		\includegraphics[width=0.49\columnwidth]{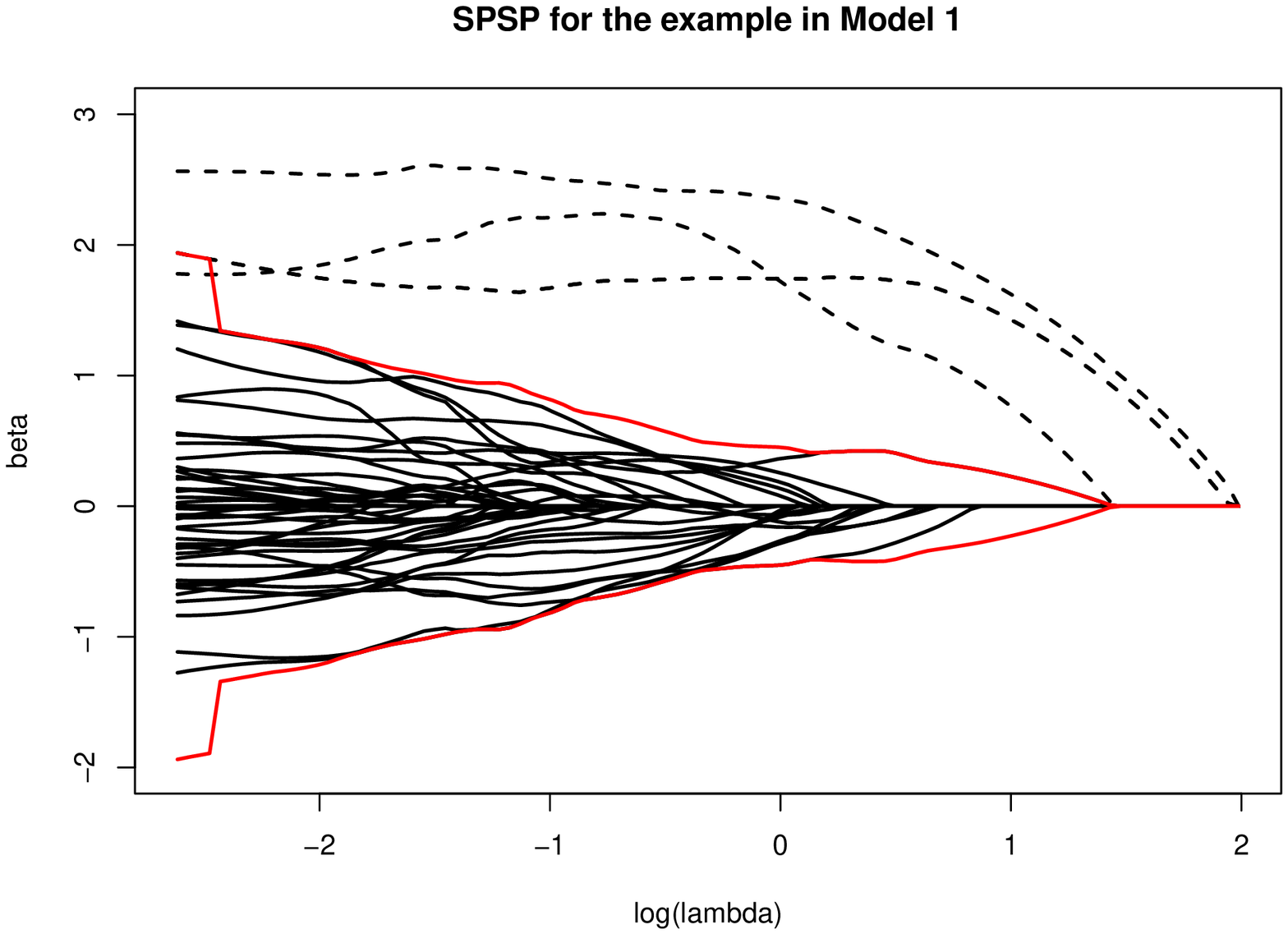}
		
		\includegraphics[width=0.49\columnwidth]{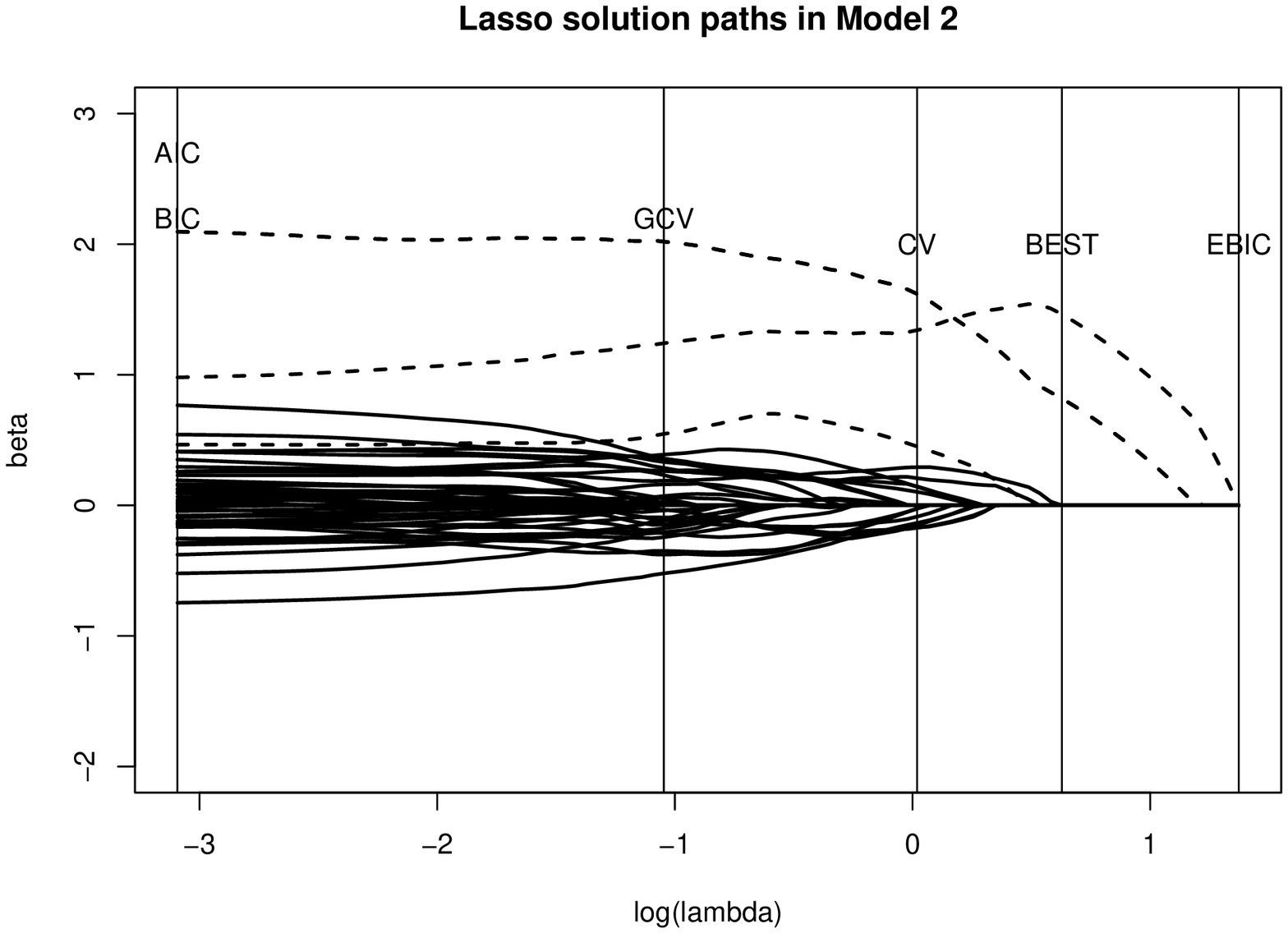}
		\includegraphics[width=0.49\columnwidth]{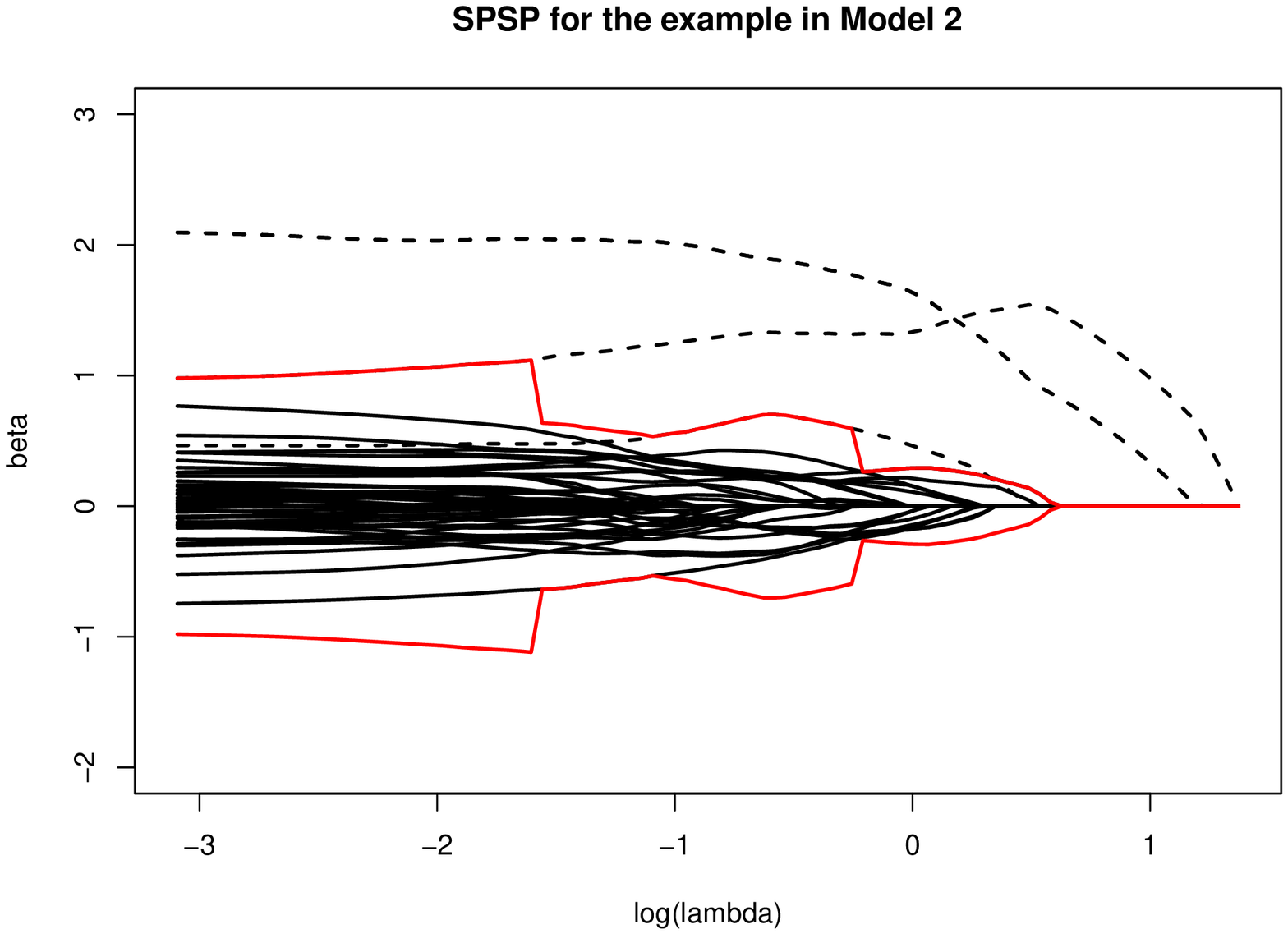}
		
		\includegraphics[width=0.49\columnwidth]{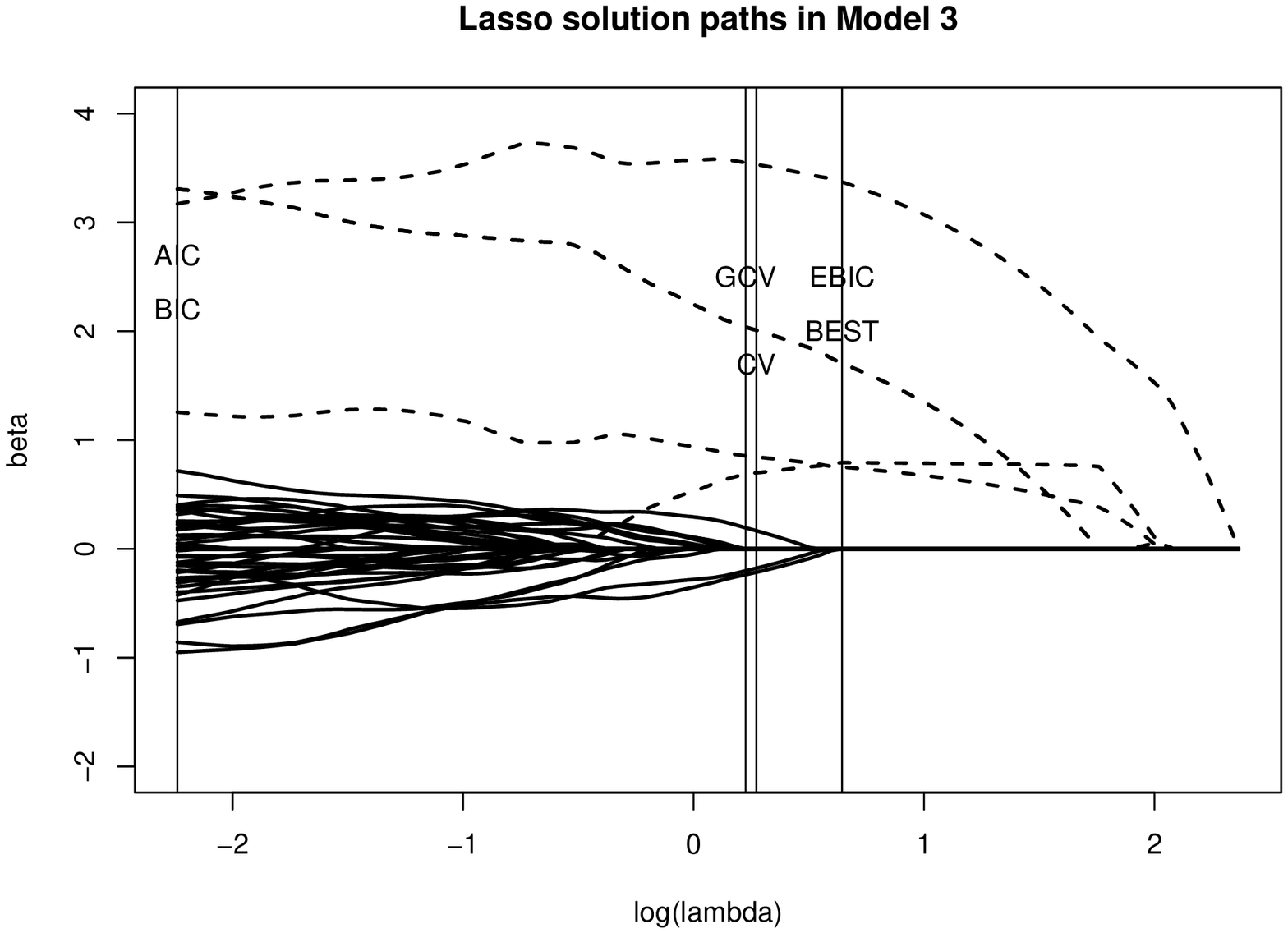}
		\includegraphics[width=0.49\columnwidth]{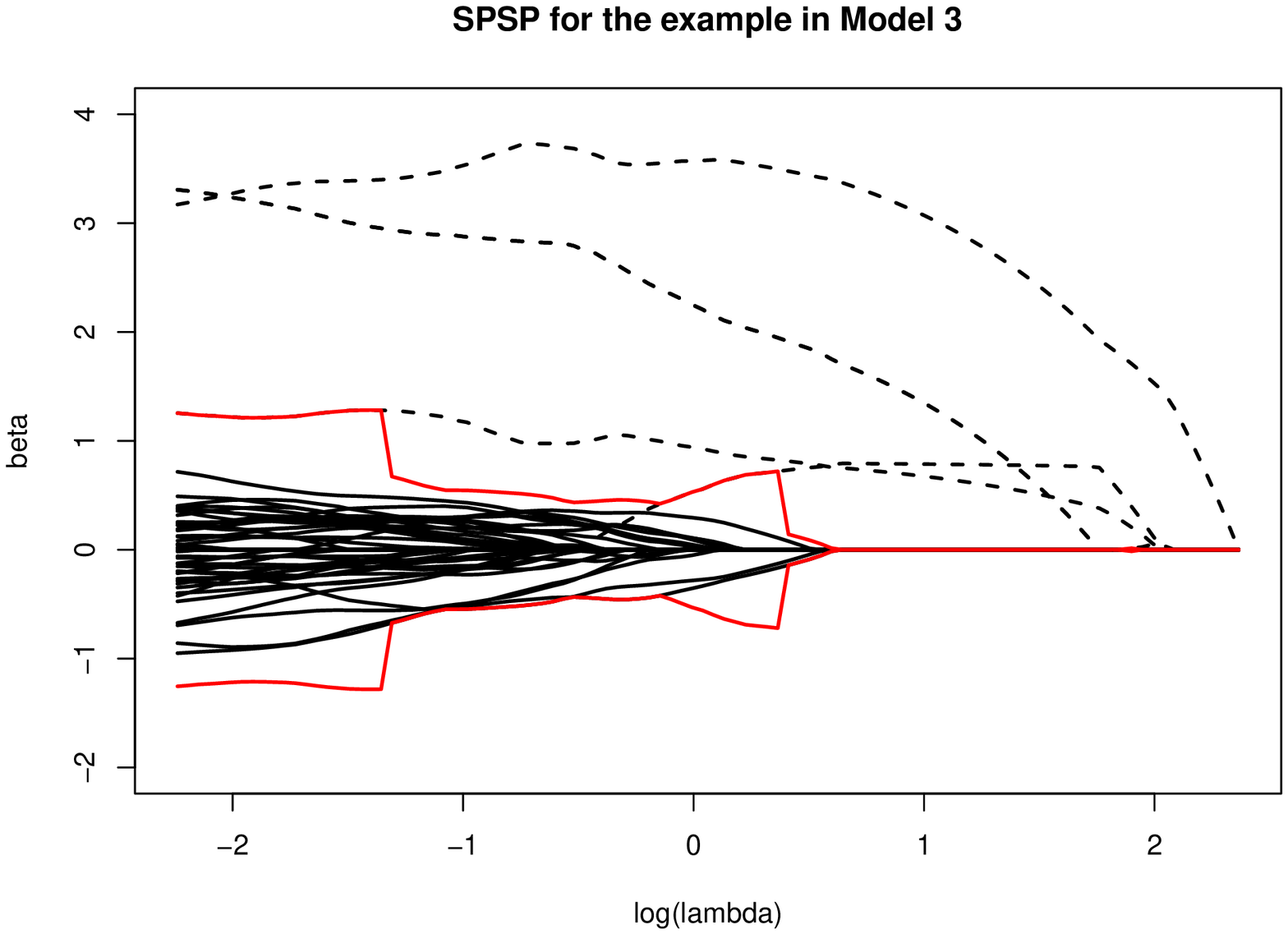}
		
		\includegraphics[width=0.49\columnwidth]{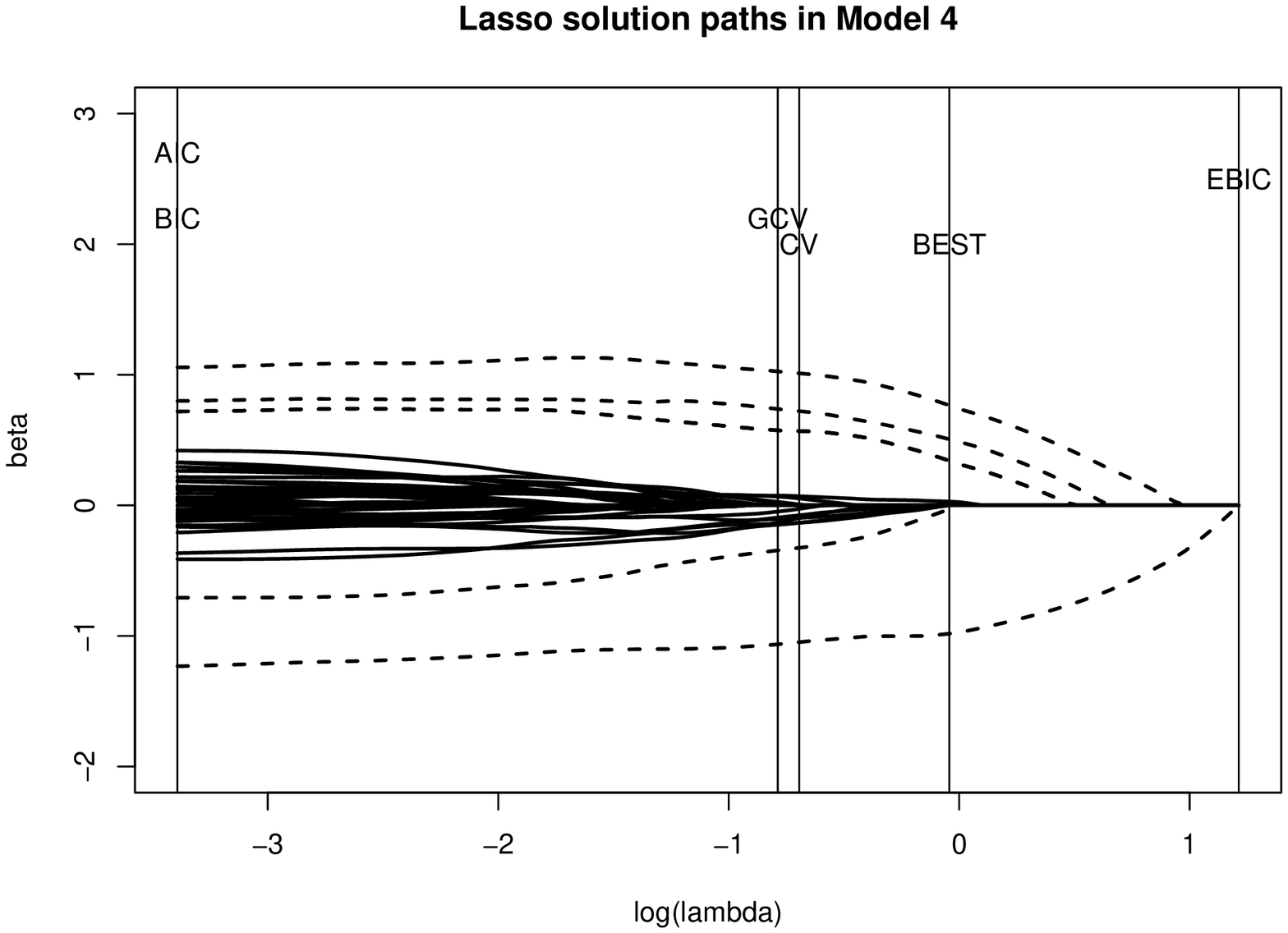}
		\includegraphics[width=0.49\columnwidth]{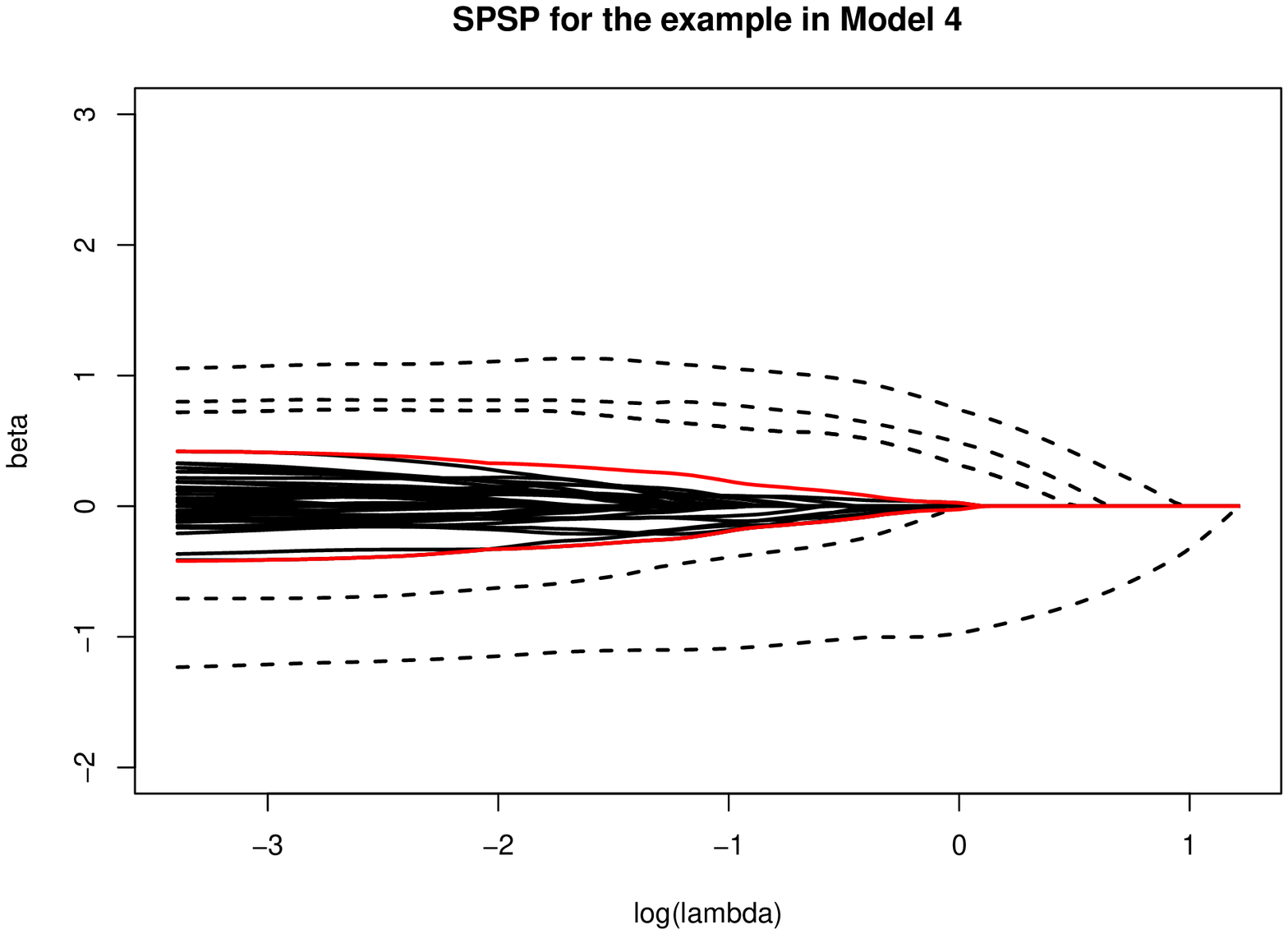}
		
		\caption{Left Figures: The lasso solution paths for one simulated examples of these four models. The dashed lines are the paths of the non-zero coefficients, while the black lines are the paths of the zero coefficients 
			The vertical lines represent the tuning parameters selected by different criteria.
			Right Figures: Partitions of the lasso solution paths of the same simulated examples using SPSP.
		}
		\label{af4}		
		
	\end{center}
\end{figure}

\section{Data Analysis}
In this section, we apply the proposed SPSP procedure along with the other selection criteria to analyze the glioblastoma (GBM) gene expression data from The Cancer Genome Atlas (TCGA) consortium (\textit{https://xenabrowser.net/datapages/}). The goal of the analysis is to identify the highly informative genes regarding the survival time of glioblastoma. 
The two gene expression data sets we obtained were both measured experimentally by the University of North Carolina TCGA genomic characterization center.  In our analysis, we use the logarithm of the survival time as the response variable. 
The sample sizes of the two data sets are 428 and 91, after removing one common sample from the two data sets. We use the larger data set as our training set and the smaller one as the testing set. 
%
%


We first screen the $17814$ genes with the training data using sure independence screening \citep{fan2008sure} to identify the $1000$ genes which are most correlated with the response \citep{wang2011random}. We then apply the proposed SPSP approach and the other methods to the training set and evaluate the prediction errors of the responses using the test data. 


Table \ref{rdt} shows the number of genes selected  in the training set and the corresponding mean prediction error  in the test set for each method.  We notice that although the SPSP for lasso selects only $4$ genes, it yields the smallest prediction error among all the penalty--criterion combinations. On the other hand, cross-validation, generalized cross-validation, AIC, and BIC for lasso  all select too many irrelevant genes and therefore produce larger prediction errors. Both the extended BIC and stability selection select fewer variables and produce larger prediction errors. In fact, the extended BIC fails to recognize any signals for any of the penalties, and stability selection fails to recognize any signals for the non-convex penalties.
%

\begin{table}
	\caption{Analysis of the glioblastoma data set}
	\label{rdt}
	\centerline{
			\begin{tabular}{ccccccccc}
				\hline
				&  & CV  & GCV & AIC & BIC & EBIC & STAB  & SPSP \\ 
				\hline
				Lasso & NUM & 37 & 157 & 275 & 2 & 0 & 3 & 4 \\ 
				& MPE & 1.38 & 1.215 & 1.644 & 1.317 & 1.302 & 1.579 & 1.143 \\ \hline
				adaLasso & NUM & 92 & 100 & 150 & 5 & 0 & 14 & 24 \\ 
				& MPE & 1.266 & 1.216 & 1.175 & 1.422 & 1.302 & 1.529 & 1.363 \\  \hline
				SCAD & NUM & 45 & 127 & 127 & 127 & 0 & 0 & 3 \\ 
				& MPE & 1.299 & 3.005 & 3.005 & 3.005 & 1.302 & 1.302 & 1.28 \\  \hline
				MCP & NUM & 23 & 113 & 113 & 76 & 0 & 0 & 13 \\ 
				& MPE & 1.277 & 2.472 & 2.472 & 1.882 & 1.302 & 1.302 & 1.502 \\ 
				\hline
	\end{tabular}}
\end{table}
%


Two genes identified by the proposed SPSP approach for lasso have been proved to be important by previous studies. MAGEC2, the Melanoma-associated antigen C2, a member of the MAGE gene family, is expressed in a wide variety of tumor types, including breast cancer  \citep{song2016cancer} and glioma \citep{sasaki2001mage}. \cite{song2016cancer} further demonstrated that the gene expression levels of MAGEC2 are highly positively correlated with the phosphorylated STAT3, the signal transducer and activator of transcription 3, in human tumor tissues, and the STAT3 activation has been shown 
to affect multiple endogenous negative regulators such as PIAS3 (protein inhibitor of activated STAT3) in glioblastoma multiforme tumors. Therefore, the results of our present analysis provide additional evidence regarding the impact of MAGEC2 on glioblastoma behavior besides the work of \cite{song2016cancer}, which may motivate scientists to conduct more experiments to explore the direct relationship between MAGEC2 and the human glioma. Moreover, the SPSP procedure additionally identifies TMEM49, also named VMP1, vacuole membrane protein 1, the gene which  encodes an important protein in the process of autophagy. \cite{ying2011hypoxia} identified the gene as the direct and functional target of MicroRNA-210 (miR-210) and investigated the associations between the expression of the gene and hepatocellular carcinoma (HCC). \cite{lai2015serum} demonstrated that miR-210 is a potential serum biomarker in the diagnosis and prognosis of glioma, and it is also worth conducting further experiments to establish the connection between the expression of this gene and glioblastoma. Finally, our SPSP approach is the only method which can identify both MAGEC2 and THEM49 without selecting more than 100 genes. 

\section{Discussion}
We have proposed a novel selection procedure for the penalized likelihood approach based on the entire solution paths. By utilizing estimators over all values of the tuning parameter, we can obtain better selection accuracy than the commonly used approach of selecting only one tuning parameter based on certain criteria. Moreover, the proposed SPSP procedure also achieves selection consistency under conditions that are substantially weaker than the irrepresentable condition, which is almost necessary under the framework currently in usage. Another advantage of SPSP is that we can now carry out selection with a strictly convex $l_2$ penalty. Although the present paper mainly focuses on feature selection for linear models, the SPSP procedure can easily be applied to most selection problems with one or more tuning parameters. In Section 3 of the appendix, we include a simulation study on estimating high-dimensional Gaussian graphical models with the SPSP approach to illustrate SPSP's potential for other selection problems.

With this study, we hope to initiate a discussion of how to better apply information contained in entire solution paths. For example, we can rank the importance of features by exploring the differences between the behaviors of the solution paths for important features and those for spurious ones. It is also possible to develop an inference procedure and quantify the uncertainty of selection results based on the entire solution paths. In addition, it might be interesting to see whether the solution paths for important features differ from those for irrelevant ones in any other manner besides the magnitude of the estimators.

\section{Code}
All source code for the simulations and real data analysis can be found on GitHub: \texttt{https://github.com/yliu433/r-spsp}.

\section{Appendix}\label{app}

\subsection{Technical proofs}\label{A1}
\begin{proof}[Proof of Theorem 1]
	By Lemma \ref{cc}, with probability at least $1-2e^{-t^2/2}$, we have
	\bee
	\hat	D(S^c_\lambda)
	\leq
	\max\{D_{i'}: i' < i_\lambda\} + \delta_\lambda
	%
	\eee
	and similarly for any $j_1\in S_\lambda,j_2\in S_\lambda^c$,
	\bee
	\begin{split}
		\left||\hat{\beta}_{j_1}|-|\hat{\beta}_{j_2}|\right|&\geq |\beta_{j_1}^\ast|-\left||\hat{\beta}_{j_1}|-|\beta_{j_1}^\ast|\right|-|\hat{\beta}_{j_2}|\\
		&\geq D_{i_\lambda}-\delta_\lambda
	\end{split}
	\eee
	
	Then with probability at least $1-2e^{-t^2/2}$,
	\bee
	\begin{split}
		\frac{\hat D(S_\lambda,S_\lambda^c)}{\hat D(S^c_\lambda)}\geq \frac{D_{i_\lambda}-2\delta_\lambda}{\max\{D_{i'}: i' < i_\lambda\} + \delta_\lambda}& > R.
	\end{split}
	\eee
	The above inequality follows immediately from $C_{under}^{i_\lambda}>R$ and $ D_{i_\lambda}>(1-\frac{R}{C_{under}^{i_\lambda}})^{-1}(1+R)\delta_\lambda.$

	To verify the partitioning rule, we have
	\bee
	\begin{split}
		\hat D(S_\lambda)\leq D_{\max}+\delta_\lambda
		%
		%
	\end{split}
	\eee
	
	Then with probability at least $1-2e^{-t^2/2}$,
	\bee
	\begin{split}
		\frac{\hat D(S_\lambda)}{\hat D(S_\lambda,S_\lambda^c)}&\leq\frac{D_{\max}+\delta_\lambda}{D_{i_\lambda}-\delta_\lambda},\\
		&\leq C+ \frac{(C+1)\delta_\lambda}{D_{i_\lambda}-2\delta_\lambda} \\
		&\leq C+ \frac{1+C}{R}=R.
	\end{split}
	\eee
	The last inequality follows from the fact that $R=1+C.$
\end{proof}

\begin{proof}[Proof of Theorem 2]
	A sufficient condition for $\hat S_\lambda \cap S^C=\emptyset$ is
	$$
	D_{\max} > R \delta_\lambda.
	$$
	Then the theorem follows immediately from Theorem \ref{compact} and the fact that  $\hat S=\cup_\lambda \hat S_\lambda.$
\end{proof}

\begin{proof}[Proof of Lemma 2] Note that 
	\begin{eqnarray*}
		&&	\|\bX\bbeta^\ast-\bX_{S}\hat\bbeta_{S}-\bX_{S^C}\hat\bbeta_{S^C}\|^2 \\
		\geq &&	\|\bX\bbeta^\ast-\bX_{S}\hat\bbeta_{S}-\bX_S\mbox{diag}(\sign(\bbeta^\ast_S)) \sign(\bbeta_S^\ast)(\bX_S^T\bX_S)^{-1}\bX^T_S\bX_{S^C}\hat\bbeta_{S^C}\|^2.
	\end{eqnarray*}
	By the irrepresentable condition, there is a $\eta>0$ such that
	\[
	\|\sign(\bbeta_S^\ast)(\bX_S^T\bX_S)^{-1}\bX^T_S\bX_{S^C}\|_{\infty} \leq 1-\eta,
	\]	
	from which it follows immediately that
	$$\|\mbox{diag}(\sign(\bbeta^\ast_S)) \sign(\bbeta_S^\ast)(\bX_S^T\bX_S)^{-1}\bX^T_S\bX_{S^C}\hat\bbeta_{S^C}\|_1 \leq \|\hat\bbeta_{S^C}\|_1(1-\eta).$$ Therefore, the identifiability condition holds if the irrepresentable condition holds.

\end{proof}

To prove Theorem 3, we introduce the following two lemmas.
\begin{lemma}
	\label{weak}
	Under $WIC(k, \kappa)$, the following inequality holds for the lasso solution $\hat{\bbeta}= (\hat\bbeta_{S}, \hat{\bbeta}_{S^C})$ with $\lambda > \lambda_0(\frac{2+2k-k\eta}{k\eta}+\kappa)$ with probability at least $1-2e^{-t^2}$,
	\[
	\|\hat{\bbeta}_{S^C}\|_1 \leq k \|\hat{\bbeta}_S\|_1.
	\]
\end{lemma}

\begin{proof}
	Since $\hat{\bbeta}= (\hat\bbeta_{S}, \hat{\bbeta}_{S^C})$ is the lasso solution, then for $$\tilde{\bbeta}_S=\arg \min_{\|\bbeta_S\|_1 \leq \|\hat\bbeta_S\|_1+ (1-\eta)\|\hat{\bbeta}_{S^c}\|_1}\| \bX\bbeta^\ast-\bX_{S}\bbeta_{S}  \|^2, $$ we have
	\bee
	\begin{split}
		&\frac{1}{n} \| \bX\bbeta^\ast-\bX_S\tilde{\bbeta}_S\| + \lambda \|\tilde{\bbeta}_S\|_1 + \frac{1}{n}2\varepsilon^T\bX(\bbeta^\ast-\tilde{\bbeta}) \\
		\geq &\frac{1}{n}\|\bX \bbeta^\ast-\bX_S\hat{\bbeta}_S-\bX_{S^C}\hat{\bbeta}_{S^C}\|+  \lambda \|\hat{\bbeta}\|_1 + \frac{1}{n}2\varepsilon^T\bX(\bbeta^\ast-\hat{\bbeta}).
	\end{split}
	\eee
	If the following inequality holds when $\|\hat{\bbeta}_{S^C}\|_1 > k \|\hat{\bbeta}_S\|_1$
	\[
	\lambda\|\hat{\bbeta}\|_1-\lambda\|\tilde{\bbeta}_S\|_1+\frac{1}{n}2\varepsilon^T\bX(\tilde\bbeta-\hat{\bbeta}) \geq 0.
	\]
	Then it follows from the identifiability condition that
	\bee
	\begin{split}
		&\frac{1}{n} \| \bX\bbeta^\ast-\bX_S\tilde{\bbeta}_S\| + \lambda \|\tilde{\bbeta}_S\|_1 + 2\varepsilon^T\bX(\bbeta^\ast-\tilde{\bbeta}) \\ 
		\leq &\frac{1}{n}\|\bX \bbeta^\ast-\bX_S\hat{\bbeta}_S-\bX_{S^C}\hat{\bbeta}_{S^C}\|+  \lambda \|\hat{\bbeta}\|_1 + \frac{1}{n}2\varepsilon^T\bX(\bbeta^\ast-\hat{\bbeta}),
	\end{split}
	\eee
	therefore either $\hat{\bbeta}_{S^c}=0$ or $\|\hat{\bbeta}_{S^C}\|_1 \leq k \|\hat{\bbeta}_S\|_1$.
	
	Because we have concluded  $P(\|\frac{1}{n}2\varepsilon^T\bX\|_{\infty}<\lambda_0)>1-2e^{-t^2}$, when $\|\hat{\bbeta}_{S^C}\|_1 > k \|\hat{\bbeta}_S\|_1,$ we have 
	\begin{eqnarray*}
		&&\lambda\|\hat{\bbeta}\|_1-\lambda\|\tilde{\bbeta}_S\|_1+\frac{1}{n}2\varepsilon^T\bX(\tilde\bbeta-\hat{\bbeta}) \\
		& \geq&  \lambda\|\hat{\bbeta}\|_1-\lambda\|\tilde{\bbeta}_S\|_1 -\lambda_0\|\tilde\bbeta-\hat{\bbeta}\|_1\\
		& \geq& \lambda \eta \|\hat{\bbeta}_{S^C}\|_1 -\lambda_0\|\tilde{\bbeta}_S\|_1-\lambda_0\|\hat{\bbeta}\|_1\\
		& \geq& \lambda \eta \|\hat{\bbeta}_{S^C}\|_1 -\lambda_0(\frac{1}{k}+1-\eta) \|\hat{\bbeta}_{S^C}\|_1 - \lambda_0(\frac{1}{k}+1) \|\hat{\bbeta}_{S^C}\|_1\\
		& = & \{\lambda\eta -\lambda_0(\frac{2}{k}+2-\eta)\}\|\hat{\bbeta}_{S^C}\|_1\\
		& >&  0.
	\end{eqnarray*}
	The last inequality follows from $\lambda > \lambda_0\frac{2+2k-k\eta}{k\eta}$.

\end{proof}

\begin{lemma}\label{small2}
	Under the weak identifiability condition with $k=\frac{2}{2s+Rs(s+1)}$, $$ P(\hat S_{\lambda} \subset S) \geq 1-2e^{-t^2},$$
	for $\lambda > \lambda_0(\frac{2+2k-k\eta}{k\eta}+\kappa).$
\end{lemma}

\begin{proof}
	Denote $\hat\beta^{\max}_{S^C}=\max\{|\hat{\beta}_j|: j \in S^C\}$, and sort the absolute values in  $\{|\hat{\beta}_j|: j \in S,|\hat{\beta}_j| \geq \hat\beta^{\max}_{S^C}\}$ in ascending order to get $ \hat{\beta}^u_{(1)}\leq  \hat{\beta}^u_{(2)}\dots \leq \hat{\beta}^u_{(d)},$ where $d \leq s$ is the cardinality of the set  $\{|\hat{\beta}_j|: j \in S,|\hat{\beta}_j| \geq \hat\beta^{\max}_{S^C}\}.$  Let $\Delta_1=\hat{\beta}^u_{(1)}-\hat\beta^{\max}_{S^C}$ and $\Delta_i=\hat{\beta}^u_{(i)}- \hat{\beta}^u_{(i-1)},$ $i=2, \dots,d,$  then $\|\hat{\beta}_S\|_1 \leq  s\hat\beta^{\max}_{S^C} + \sum_{i=1}^d i\Delta_i. $ Therefore by Theorem \ref{weak} $$\hat\beta^{\max}_{S^C} \leq \|\hat{\beta}_{S^C}\|_1 \leq ks\hat\beta^{\max}_{S^C}+k \sum_{i=1}^d i\Delta_i \leq ks\hat\beta^{\max}_{S^C}+k\frac{s(s+1)}{2}\Delta_{\max},$$
	where $\Delta_{\max}$ is the maximum value of $\Delta_1, \dots, \Delta_d.$ It follows when $k=\frac{2}{2s+Rs(s+1)}$ that
	$$ \hat\beta^{\max}_{S^C} \leq \frac{ks(s+1)}{2(1-ks)}\Delta_{\max}=\frac{1}{R} \Delta_{\max}.$$

\end{proof}

\begin{proof}[Proof of Theorem 3]
	The theorem follows directly from Theorem \ref{small} and Lemma \ref{small2}.
\end{proof}

\subsection{Simulation studies on sensitivity SPSP to the value of $R$}\label{A2}
The constant $R$ in the proposed SPSP algorithm controls the order of the magnitude in the partitioning rule. Here we examine the sensitivity of the SPSP procedure with respect to the choice of the constant $R$. Note that in the numerical studies we conducted, these constants are data adaptive rather than given arbitrarily. 


In particular, we select a sequence of the numbers from $1$ to $10$ with the increment $0.5$, as the candidates of the constant $R$. We use each number as the constant $R$ in the proposed SPSP algorithm on lasso for all the models, and record the means of the false positive rates ($FPR = FP/$\# of true zero features) and the false negative rates ($FNR = FN/$\# of true nonzero features) for each value of $R$.

The results are shown in Figure \ref{SC}. We observe that the means of the FPR and FNR are relatively stable across different choices of the constant $R$. Note that the vertical line in each graph represents the mean of the values of $R$ estimated from the data set. It illustrates that the selection results of the SPSP algorithm are not so sensitive to the choice of the constant $R$ as long as $R$ stays within a reasonable range.

\begin{figure}
	\vspace{6pc}
	\begin{center}
		\caption{The mean of the FPR and FNR over $500$ replicates over different choices of the constant $R$ in SPSP on the Lasso. The vertical lines in the graphs are the average values of the $R$ estimated from the data set.}
		\centerline{\includegraphics[width=\columnwidth]{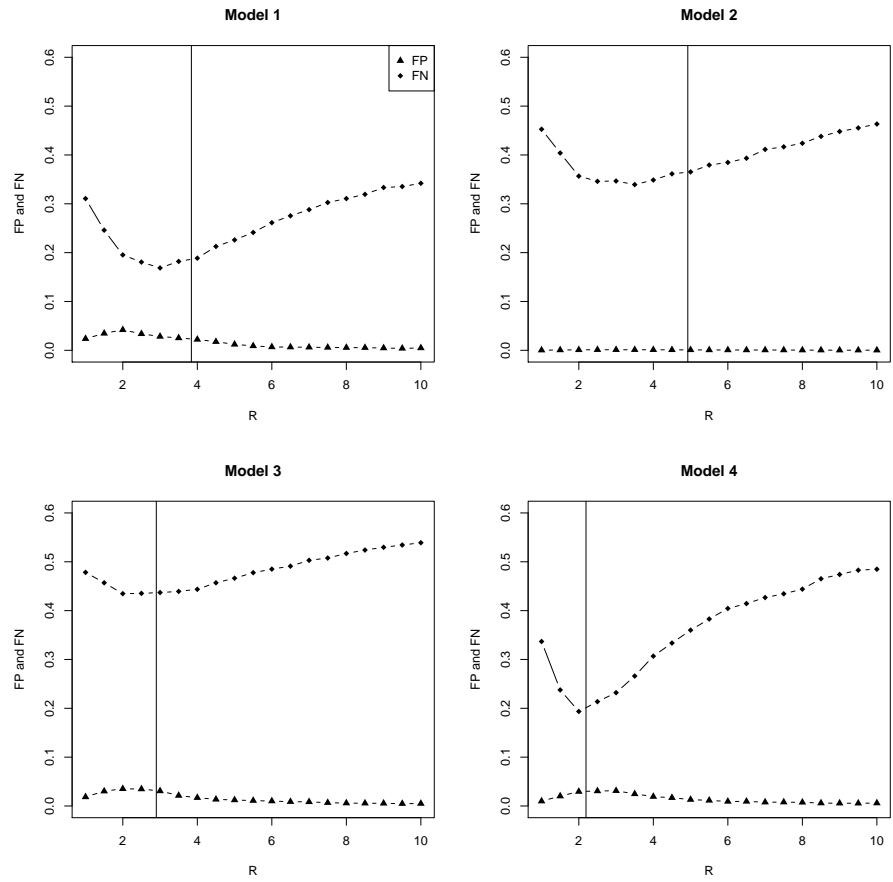}}
		\label{SC}
	\end{center}
\end{figure}

\subsection{Simulation study on the Gaussian graphical model}\label{A3}

Besides feature selection for linear models, the SPSP procedure introduced in the paper can be widely applied for other selection problems under the framework of penalized likelihood estimation. Here we simply present a simulation study to illustrate the performance of the SPSP algorithm for the Gaussian graphical model.

The data is simulated from
$\mathbf{X}\sim N_p(\mathbf{0},\mathbf{\Sigma})$, where the inverse of the covariance matrix is set as
$(\mathbf{\Sigma}^{-11})_{j,j}=1,(\mathbf{\Sigma}^{-11})_{j,j+1}=0.5,(\mathbf{\Sigma}^{-11})_{j+1,j}=0.5,j=1,\cdots,p/4,$ and zero otherwise. We set $p=100$ and $n=50$ in the simulation. This example is a AR$(1)$ model, which has been used by \cite{friedman2008sparse} and \cite{yuan2007model} for the numerical study of the graphical lasso.

We  compare the performance of the proposed SPSP algorithm with the graphical models selected by BIC and the EBIC. Here the details about using BIC and the EBIC to choose the tuning parameter in the graphical models are described in \cite{foygel2010extended}. We apply the \textbf{R} package \textit{glasso} to solve the graphical lasso estimators and apply the package \textit{qgraph} to select the graphical lasso models by BIC and the EBIC. Note that the grid of the tuning parameters in the simulation is generated automatically by the function \textit{glassopath}.

\begin{table}[h] 
	\caption[The mean of FP, FN values of the SPSP algorithm, BIC, and the EBIC over $100$ replicates. ]{The mean of FP, FN values of the SPSP algorithm, BIC, and the EBIC over $100$ replicates (Standard Error in the parentheses). The true model has $25$ nonzero dependencies and $4925$ zero dependencies.}
	\label{glasso}
	\centerline{\tabcolsep=3truept
		\begin{tabular}{rccc}
			\hline
			& SPSP & BIC  & EBIC \\
			\hline
			FP                       & $19.31$ & $116.56$ & $0$ \\
			& $(2.48)$ & $(3.2)$  & $(0)$ \\
			\hline
			FN                      &  $2.50$   & $0$   & $25$  \\
			& $(0.80)$ & $(0.0)$ & $(0)$ \\
			\hline
	\end{tabular}}
\end{table}

We report the mean and the standard error of the number of the false positives (FP), the number of the false negatives (FN) of the SPSP algorithm, BIC and the EBIC over $100$ replicates in Table~\ref{glasso}. We observe that the BIC tends to include too many zero dependencies (high FP value) while the EBIC missed all the nonzero dependencies (high FN value) in the model. Compared with the results of these two criteria, the SPSP algorithm has a much better performance in terms of selection accuracy, which selects most of the nonzero dependencies without adding many zero dependencies in the model.

\bibliographystyle{imsart-nameyear}
\bibliography{spsp_ejs_ref}

\end{document}